\documentclass[11pt]{scrartcl}
\usepackage[a4paper, total={16cm, 24cm}]{geometry}

\usepackage{multicol}
\usepackage{tocloft}

\usepackage[small]{caption}
\usepackage{subcaption}

\usepackage{algorithm}
\usepackage[table,dvipsnames]{xcolor}
\usepackage{todonotes}
\usepackage{pdflscape}
\usepackage{adjustbox}

\usepackage{natbib}
\usepackage[pagebackref]{hyperref}
\hypersetup{
	pdfencoding=auto,
	psdextra,
	colorlinks=true,
	citecolor=green!50!black,
	linkcolor=red!60!black,
	urlcolor=purple!70!black
}

\renewcommand{\epsilon}{\varepsilon}

\usepackage{authblk}

\usepackage{url}
\usepackage{microtype}
\usepackage[utf8]{inputenc}
\usepackage{graphicx}
\usepackage{amsmath}
\usepackage{amsthm}
\usepackage{amssymb}
\usepackage{mathtools}
\usepackage{nicefrac}
\usepackage{booktabs}
\usepackage[inline]{enumitem}
\usepackage{tabulary}
\usepackage{tabularx}
\usepackage{multirow}

\usepackage{wrapfig}

\usepackage[nameinlink]{cleveref}
\usepackage{stmaryrd}
\usepackage{thm-restate}
\usepackage{tikz}
\usetikzlibrary{arrows.meta}
\usepackage{tcolorbox}

\newtheoremstyle{sfthm}
	{\topsep}
	{\topsep}
	{\itshape}
	{}
	{\sffamily\bfseries}
	{}
	{.5em}
	{}
\theoremstyle{sfthm}
\newtheorem{theorem}{Theorem}[section]
\newtheorem{lemma}[theorem]{Lemma}
\newtheorem{proposition}[theorem]{Proposition}

\newtheoremstyle{sfdef}
{\topsep}
{\topsep}
{}
{}
{\sffamily\bfseries}
{}
{.5em}
{}
\theoremstyle{sfdef}

\newtheorem{example}[theorem]{Example}

\DeclareMathOperator\supp{supp}
\DeclareMathOperator\closure{cl}

\newcommand{\leLR}{\le_{\mathrm{lr}}}
\newcommand{\leST}{\le_{\mathrm{st}}}
\newcommand{\Exp}{\mathbb{E}}

\AddToHook{cmd/appendix/before}{%
  \crefalias{section}{appendix}%
  \crefalias{subsection}{appendix}%
}

\providecommand{\theHtheorem}{}
\renewcommand{\theHtheorem}{\thesection.\arabic{theorem}}

\title{Fair Division of Lotteries}
\author{Ayumi Igarashi}
\author{Yeeseok Oh}
\affil{The University of Tokyo}
\date{\vspace{-3em}}

\begin{document}

\maketitle

\begin{abstract}
\begin{center}
    \textbf{\textsf{Abstract}} \smallskip
\end{center}
We study fair division of uncertain items, where the utility of each item is a random variable whose realization is observed only after the allocation is determined. This uncertainty creates a distinction between ex-ante fairness, evaluated before utilities are realized, and ex-post fairness, evaluated after realization. 
We examine to what extent these two notions are compatible. 
We first identify conditions under which ex-ante envy-free allocations maximize ex-post envy-freeness probability. These include two-agent instances where utilities toward the same item follow the same distribution from a convolution-power family. For an arbitrary number of agents, we establish analogous optimality results when the common generator is non-atomic or log-concave.
Beyond convolution-power families, we derive approximation guarantees for ex-ante envy-free allocations relative to the optimal ex-post envy-freeness probability under log-concave or symmetric utilities. We also study the computational complexity of evaluating ex-post envy-freeness probabilities, establishing hardness results as well as pseudo-polynomial-time algorithms.
\end{abstract}

 \section{Introduction}
 
The classical problem of fair division studies the problem of allocating resources to agents in a fair manner. This problem arises naturally in real-world settings: dividing estates during inheritance, drafting amateur players in professional sports leagues, or assigning papers to program committee (PC) members. With its broad applications, fair division has been studied extensively in mathematics, economics, and computer science \citep{BramsTaylor1996FairDivision,Moulin2003FairDivision,AmanatidisEtAl2022FairDivisionSurvey}.

In classical fair division, agents express their utilities for items using fixed, real numbers. However, in real-world settings, agents are often unsure about how good an item is.
In sports drafts, teams select amateurs based on their expectations for the player to perform in the pro league. Still, a highly expected amateur player could play poorly in the pro league or get injured, failing to meet the high standards he had. Conversely, it is not an odd story that a player signed for a low salary turns out to be a star.
In reviewer systems, PC members often bid for the papers they want to review, based on the limited information of titles and abstracts. 
Again, papers that seemed interesting can turn out to be of poor quality, or a good piece of work. 

These scenarios raise the following question:
``Is an allocation that appears fair before the uncertainties are realized actually fair after the realization?''

\paragraph{Our Contributions}
We study fair division of uncertain items, where the utilities of items have exogenous uncertainty that cannot be resolved prior to allocation of the items. We evaluate how fair 
an allocation is, by considering the time \emph{before} and \emph{after} the uncertainty is resolved. As a concrete example, consider the following scenario.

\begin{example}\label{ex:running-ex}
	Consider two fishers, Tina and Sam, who specialize in tuna and salmon, respectively. Each fisher cares only about the amount of their preferred fish.
    They are dividing three fishing grounds. 
    Two of the grounds are well understood and each yields $5$ tons of both tuna and salmon.
    The third ground is unexplored. For each species independently, it contains either $1000$ tons of it with probability $0.01$ or no fish at all otherwise.
    
	One possible allocation is giving Tina the two known grounds and Sam the unexplored ground. Since the expected amount of each fisher's preferred fish is $10$ tons, the fishers may perceive the allocation as fair. 
    Sam, however, with a probability of $0.99$, gets zero salmon, envying Tina.
    Tina also would envy Sam with a probability of $0.01$ if Sam's ground turns out to be full of tuna.
    Hence, this allocation ends up being \emph{envy-free} only with a probability of $(1-0.99) \cdot (1-0.01) = 0.0099$.
	
	Another option is allocating the two fishers each a well-known ground and the unexplored to one of the two, say Sam. At the time of the allocation, Tina would envy Sam. However, with a probability of $0.99$, the ground will turn out to have no tuna, and Tina would not envy Sam. Since Sam would never envy Tina, this allocation turns out to be envy-free with probability $0.99$.
\end{example}

The example illustrates that assessments of fairness can differ \emph{before} and \emph{after} the lotteries are realized. 
Our central question is to what extent an allocation that performs well \emph{ex ante} (before lotteries are realized) can also perform well \emph{ex post} (after realizations). We will focus on the commonly used fairness notion of \emph{envy-freeness} (EF)~\citep{foley1967resource}. 
Since \Cref{ex:running-ex} already shows that an ex-ante EF allocation may not be ex-post EF,
we think of allocations that perform well ex post as those that maximize the probability of being EF ex post.
Our results are summarized as follows. 

\begin{itemize}

\item In \Cref{sec:EFcompat}, we show that an ex-ante EF allocation maximizes the ex-post EF probability when there are two agents and each item has the same utility distribution across agents drawn from a common convolution-power family. 
Further, we show that the result extends to an arbitrary number of agents when the family's generator is either non-atomic or log-concave.

\item In \Cref{sec:approx}, we consider general distributions beyond con\-vo\-lu\-tion-power families. While an ex-ante EF allocation does not necessarily maximize ex-post EF probability in this regime, we show that under log-concavity or symmetry, an ex-ante EF allocation provides a constant probability of being EF ex post, for a fixed number of agents.

\item In \Cref{sec:complexity}, we turn to the problem of computing the ex-post EF probability of a given allocation. We show that this is weakly $\#$P-hard. Still, we show the probabilities can be computed in pseudo-polynomial time using dynamic programming. 

\item In \Cref{sec:identical}, we consider a different, but important case where valuations to the same good are identical. We specify how the case differs from our central model and show that most of the results transfer to this setting.

\end{itemize}

\paragraph{Related Work}
\citet{aloisio2024approximately} study a model closely related to ours, with additive random valuations that are realized only after the items have been allocated. They consider the problem of minimizing expected maximum envy and develop both deterministic and randomized allocation algorithms, with the latter providing ex-ante fairness guarantees. Our focus is on fixed allocations and the relationship between EF under expected utilities and the probability of allocations being EF after realization. In particular, we identify distributional conditions under which ex-ante EF allocations, when they exist, maximize ex-post EF probability.

Other work on allocating indivisible resources considers different forms of uncertainty. \citet{LumetEtAl2012} study a model in which each indivisible item can realize to a good or worthless, with their conditions revealed only after allocation. They examine ex-ante and ex-post welfare, as well as the probability of satisfying proportionality. \citet{AzizEtAl2024UncertainHouse} consider the problem of maximizing EF probability under uncertain ordinal preferences in house allocation. Their model assigns one item to each agent, whereas ours allows bundles evaluated through additive random cardinal utilities.

Random utilities also feature in \emph{asymptotic fair division}~\citep{dickerson2014computational,manurangsi2020when,manurangsi2021closing,manurangsi2017asymptotic,bai2022envyfree,manurangsi2025weighted}, which identifies conditions under which randomly generated instances admit fair allocations with high probability. In these models, the allocation may depend on the realized utility profile. We fix an allocation before utilities are realized and evaluate its probability of being EF afterward.

A different role for randomness arises in \emph{best-of-both-worlds fair division}~\citep{aziz2024best,hoefer2024best,guo2023first}. This literature seeks lotteries over allocations that combine ex-ante fairness with ex-post fairness guarantees for every allocation in their support. Here, the mechanism randomizes the allocation while utilities are fixed; in our model, the allocation is fixed while utilities are subject to exogenous uncertainty.

Related questions have also been studied in rent division and divisible-resource allocation. \citet{PetersEtAl2022} study rent division robust to perturbations of reported room utilities, including maximizing the probability of EF under a distribution over perturbations. \citet{BuermannEtAl2020} consider a homogeneous divisible resource whose available quantity is uncertain and allocations contingent on the realized quantity. They show that ex-post EF in every realization implies ex-ante EF, while the converse may fail, and optimize ex-ante welfare subject to ex-ante EF.

\section{Preliminaries}\label{sec:prelim}

For a positive integer $t$, let $[t]=\{1, 2, \ldots, t\}$.
We start by introducing preliminaries on probability-theoretic notions.
Throughout, all random variables are real-valued and have finite mean. Since a random variable determines a probability distribution, when no confusion arises, we use random variables and distributions interchangeably. For two random variables $X$ and $Y$, we write $X \stackrel{d}{=} Y$ if they follow the same distribution, i.e., $\Pr(X \le t) = \Pr(Y \le t)$ for all $t \in \mathbb{R}$. If an event happens with probability $1$, we say that it happens \emph{almost surely}, abbreviated by \emph{a.s.}

A random variable $X$ is \emph{non-atomic} if its distribution has no atoms, i.e.,
$\Pr(X=t)=0$ for all $t \in \mathbb{R}$.
A random variable $X$ is \emph{nonnegative} if $X\ge 0$ a.s.
A random variable $X$ is \emph{symmetric} if its distribution is symmetric around its mean, i.e., $X-\Exp[X] \stackrel{d}{=} \Exp[X]-X$.

For a nonnegative function $f$ on $\mathbb{R}$ or $\mathbb{Z}$,
\[
\supp(f)=
\begin{cases}
    \closure({\{x\in\mathbb{R}:f(x)>0\}})
        & \text{if } f \text{ is defined on }\mathbb{R},\\[1mm]
    \{k\in\mathbb{Z}:f(k)>0\}
        & \text{if } f \text{ is defined on }\mathbb{Z},
\end{cases}
\]
where $\closure$ is the closure operator.
If a random variable $X$ has density or probability mass function (pmf) $f_X$, we write $\supp(X)=\supp(f_X)$.

A useful tool for random variables is the notion of characteristic functions. For a random variable $X$, its \emph{characteristic function} is defined as $\phi_X(t) = \Exp[e^{itX}]$. We omit the subscript $X$ if it is clear from context. It is a standard fact that $\phi_X(t)$ is well-defined for any real $t$, and the following atom inversion formula holds.
\begin{lemma}\label{lem:atom-inversion}
	Let $X$ be a random variable. Then
	$$\Pr(X=x) = \lim_{T \to \infty}\frac{1}{2T} \int_{-T}^T e^{-itx}\phi_X(t) dt.$$
\end{lemma}

The \emph{Poisson distribution with rate $\lambda \in (0, \infty)$} is the distribution with pmf $p(x)=e^{-\lambda}\frac{\lambda^x}{x!}$ for $x \in \mathbb{Z}_{\ge 0}$,
and is denoted by $\mathrm{Pois}(\lambda)$.
For a random variable $X \sim \mathrm{Pois}(\lambda)$, $\Exp[X]=\lambda$ holds.
For independent random variables $X \sim \mathrm{Pois}(\lambda_X)$ and $Y \sim \mathrm{Pois}(\lambda_Y)$, $X + Y \sim \mathrm{Pois}(\lambda_X + \lambda_Y)$ holds.

Poisson distributions are a primary model for the number of occurrences of an event in a fixed time or space. Examples include daily events like customer service requests arriving in an hour or the number of fish caught in a fishing ground for a day, and even rare events like the number of large meteors striking the Earth in a year.

We now define an important class of probability distributions that includes the Poisson distribution: \emph{log-concave} distributions. They form a large class in \emph{unimodal} distributions, both continuous and discrete, whose density or mass is single-peaked. Other examples include the uniform, Gaussian, exponential, Bernoulli, and binomial distributions.

A nonnegative real function $f: \mathbb{R} \to \mathbb{R}_{\ge 0}$ is \emph{log-concave} if its support is an interval, and it satisfies $f(\lambda x + (1-\lambda)y) \geq f(x)^\lambda f(y)^{1-\lambda}$ for all $x,y \in \mathbb{R}$ and $0<\lambda<1$.
A nonnegative discrete function $f:\mathbb{Z} \to \mathbb{R}_{\ge 0}$ is \emph{log-concave} if its support is an integer interval, and $f(k)^2 \ge f(k-1)f(k+1)$ holds for all $k \in \mathbb{Z}$.

A continuous random variable is \emph{log-concave} if it admits a log-concave density. An integer-valued random variable is \emph{log-concave} if its pmf is log-concave.
In both cases, it is known that the cumulative distribution function (cdf) of a log-concave distribution is also a log-concave function.
By saying a random variable is log-concave, we indicate it is either a continuous log-concave random variable or an integer-valued log-concave random variable.

By definition, if $X$ is log-concave, then so is $-X$. Furthermore, it is closed under convolution in both continuous and discrete settings.
For continuous functions $f,g:\mathbb{R}\to\mathbb{R}_{\ge 0}$, the convolution $f * g$ is defined as $(f*g)(x)=\int_{\mathbb{R}} f(u)g(x-u) du$. For discrete functions $f,g:\mathbb{Z}\to\mathbb{R}_{\ge 0}$, $(f*g)(x)=\sum_{u\in\mathbb{Z}} f(u)g(x-u)$. If $X$ and $Y$ are independent continuous [integer-valued] random variables with densities [pmfs] $f$ and $g$, respectively, then $X + Y$ has density [pmf] $f*g$.

\begin{lemma}\label{lem:convolution}
	Let $f$ and $g$ be log-concave functions on $\mathbb{R}$ (resp. $\mathbb{Z}$).
	If the convolution $f * g$ is finite-valued, it is also log-concave. Equivalently, the sum of independent continuous (discrete) log-concave random variables is also continuous (discrete, resp.) log-concave.
\end{lemma}
\begin{proof}
	The continuous case was proved in the context of P\'olya frequency functions by
	\citet{schoenberg1951polya}. The discrete case follows from
	\citet{keilson1971some}. For a modern survey of log-concave functions and
	distributions, see \citet{saumard2014logconcavity}.
\end{proof}

Given a distribution $\mu$, we define the \emph{$k$-fold convolution power} $\mu^{*k}$ recursively as convolving $\mu$ with itself $k$ times. 
The \emph{convolution-power family generated by $\mu$} is the collection $\mathcal{C}_\mu$ of distributions expressed by finite convolution powers of $\mu$, i.e., $\mathcal{C}_\mu = \{\mu^{*k} : k \ge 1\}$. 
In other words, it contains the distributions of finite i.i.d. sums of a random variable drawn from $\mu$. 
By \Cref{lem:convolution}, if $\mu$ is log-concave, so are distributions in $\mathcal{C}_\mu$.



\section{Model}\label{sec:model}
\paragraph{Fair Division of Lotteries}
Let $N=[n]$ be the set of \emph{agents} and $M=[m]$ be the set of \emph{items}. The subsets of $M$ are referred to as \emph{bundles}. Each agent $i$ has a \emph{utility} $u_i(j)$ for each item $j \in M$ which is a nonnegative random variable drawn from a distribution $\mathcal{D}_{ij}$. 
Since each item induces random utilities, we may refer to such items as \emph{lottery items}.
We assume all utilities have finite mean $\Exp[u_i(j)] < \infty$.
We assume that all agent-item utilities are mutually independent, that is, the random variables $\{u_i(j)\}_{i \in N, j \in M}$ are mutually independent.

An \emph{allocation} $A=(A_1, \ldots, A_n)$ is a partition of the items, i.e., $A_i \cap A_{i'} = \emptyset$ for all $i \ne i'$, and $\cup_{i=1}^n A_i = M$. We consider \emph{additive} valuations. An agent $i$ receives utility of $u_i(B)=\sum_{j \in B}u_i(j)$ upon receiving the bundle $B \subseteq M$. Note that since each $u_i(j)$ is a random variable, $u_i(B)$ is also a random variable.

Various criteria and axioms that evaluate how good an allocation is have been proposed in fair division literature. Fix an allocation $A$. For a deterministic utility profile, the \emph{utilitarian social welfare} of $A$ is the sum of utilities of each agent $\sum_{i\in N}u_i (A_i)$. The \emph{Nash social welfare} of $A$ is the product of the utilities $\prod_{i \in N} u_i(A_i)$. 
We say $A$ is \emph{envy-free} (EF) if $u_i(A_i) \ge u_i(A_j)$ holds for all $i,j \in N$.
That is, every agent values her own bundle at least as good as any other's bundle.

\paragraph{Ex-Ante and Ex-Post Problems}
Let $\Pi$ be a property of an allocation.
For each agent $i$ and item $j$, let $U_i(j)=\Exp[u_i(j)]$. An allocation is \emph{ex-ante $\Pi$} if it satisfies $\Pi$ under the deterministic valuation profile $U$; the ex-ante value of an objective $f$ is defined analogously by evaluating $f$ under $U$. For an allocation $A$, its \emph{ex-post $\Pi$ probability} is $\Pr(A\text{ satisfies }\Pi)$, and its ex-post $f$ value is $\Exp[f(A)]$ where probabilities and expectations are taken under randomness (of utilities) in $A$.

In \Cref{ex:running-ex}, the first allocation was ex-ante EF and had ex-post EF probability $0.0099$. The latter allocation was not ex-ante EF, but had ex-post EF probability $0.99$.


Our central question is whether an ex-ante solution also guarantee optimal or approximately optimal an ex-post fairness. 
As mentioned before, there are different perspectives of fairness captured by the ex-ante and ex-post notions, respectively. It is therefore desirable to satisfy both.
A further motivation for the above question becomes apparent in \Cref{sec:complexity}. We show that computing an equivalent deterministic fair division instance for the ex-ante problem can be done efficiently, whereas considering the ex-post problem directly is hard. If an ex-ante solution guarantees to be an ex-post solution, we do not have to consider the ex-post problem directly, avoiding direct optimization of the ex-post objective.

Let us formally state the problem.
Given a property $\Pi$, we consider whether an ex-ante $\Pi$ allocation (if one exists) maximizes the ex-post probability $\Pr(A\text{ satisfies }\Pi)$. Similarly, given a function $f$, does an ex-ante $f$ maximizing allocation maximize the ex-post $f$ value?
We first see that the welfare objectives behave exactly as in the deterministic expected instance.

\begin{proposition}\label{prop:sw-ftns}
	An ex-ante utilitarian social welfare (or Nash social welfare) maximizing allocation maximizes ex-post utilitarian social welfare (or Nash social welfare, resp.).
\end{proposition}
\begin{proof}
    For every allocation, ex-ante and ex-post utilitarian welfare coincide by linearity of expectation. The same holds for Nash welfare because the agents' own-bundle utilities $u_1(A_1), \ldots, u_n(A_n)$ are mutually independent.
\end{proof}


\section{Achieving Ex-Post Fairness with Ex-Ante Solutions}\label{sec:EFcompat}
We study conditions under which an ex-ante EF allocation (if one exists) maximizes the ex-post EF probability. \Cref{ex:running-ex} shows that even for two agents, this is not the case; the first allocation is an ex-ante EF allocation, but it has strictly smaller ex-post EF probability than that of the latter.

However, the example modeled the unexplored ground with an extreme distribution of zero or a thousand tons of fish. A natural alternative is to model catches over a fixed period using Poisson distributions, as in the following example.\footnote{Poisson distributions model counts of events within a fixed amount of time, space, volume, or distance.}

\begin{example}\label{ex:fishing-poisson}
    Consider two fishers Tina and Sam dividing $m$ fishing grounds. 
    Each fishing ground $j$ has daily yield of tuna and salmon, both following $\mathrm{Pois}(\lambda_j)$, independently.
    Does an ex-ante EF allocation of the grounds maximize ex-post EF probability?
\end{example}

It turns out that, in the Poisson setting, ex-ante EF allocations maximize the ex-post EF probability for any number of agents; see \Cref{cor:poisson}.
We identify broader sufficient conditions beyond Poisson distributions. Throughout the section, each item has the same utility distribution across agents. For two agents, it suffices that these utility distributions belong to a common convolution-power family. For more than two agents, this conclusion can fail, but it holds when the common generator is non-atomic or log-concave.




\begin{theorem}\label{thm:n2homogeneous}
    Let $n=2$ and $\mu$ be a distribution. Assume $\mathcal{D}_{1j} = \mathcal{D}_{2j} \in \mathcal{C}_\mu$ for all $j$. 
    Then, an ex-ante EF allocation maximizes ex-post EF probability.
\end{theorem}
\begin{proof}
	Let $X$ be a random variable drawn from $\mu$. Since $X \ge 0$ a.s., if $\Exp[X] = 0$, then $X=0$ a.s.
	Then, $u_i(j) = 0$ a.s. for any $i,j$, and any allocation is both ex-ante EF and ex-post EF a.s., which implies the result. Assume $\Exp[X] > 0$. For each item $j$, let $k_j$ be the positive integer such that $u_i(j) \stackrel{d}{=} S_{k_j}$, where $S_{k_j}$ is the sum of $k_j$ independent copies of $X$. We may use $S'_{k_j}$ as a random variable with the same distribution that is independent of $S_{k_j}.$

	For a subset $M' \subseteq M$, let $k(M') = \sum_{j \in M'} k_j$, and let $K = k(M)$.
	We will use this notation throughout the section.
	Note that since $\Exp[u_1(M')] = \Exp[u_2(M')]$ for any $M' \subseteq M$, an allocation $A$ is ex-ante EF if and only if $\Exp[u_1(A_1)] = \Exp[u_2(A_2)]$. Or equivalently, $k(A_1) = k(A_2) = K/2$.

    Now observe that $u_1 (A_1)$ and $u_2 (A_1)$ are independent copies of $S_{k(A_1)}$. Similarly, $u_1 (A_2)$ and $u_2 (A_2)$ are independent copies of $S_{k(A_2)}=S_{K-k(A_1)}$.
    Hence, the probability of $A$ being ex-post EF is
    \begin{align*}
        &\Pr(u_1(A_1) \ge u_1 (A_2) , u_2 (A_1) \le u_2 (A_2) ) \\ 
        &=  \Pr(u_1(A_1) \ge u_1 (A_2) ) \Pr( u_2 (A_1) \le u_2 (A_2) ) \\
        &= \Pr(S_{k(A_1)} \ge S_{K-k(A_1)}') \cdot \Pr(S_{k(A_1)} \le S_{K-k(A_1)}'),
    \end{align*}
    where the first identity follows from the independence of valuations.

    Let us write $p_k^- = \Pr(S_{k} < S_{K-k}')$, $p_k^0 = \Pr(S_{k} = S_{K-k}')$, and $p_k^+ = \Pr(S_{k} > S_{K-k}')$. Then we have $p_k^- + p_k^0 + p_k^+ = 1$ and $p_k^- = p_{K-k}^+$ from symmetry. 
    We claim that $p_k^0$ is maximized when $k=K/2$.
    
    Recall that $\phi_X (t) = \Exp[e^{itX}]$ is the characteristic function of $X$. 
    The characteristic function of $S_k - S_{K-k}'$ (denoted by $\phi_k$) is given by
    \[
    \begin{aligned}
    \phi_k (t)
    &= \Exp[e^{it{(S_k - S_{K-k}')}}] \\
    &= \Exp[e^{it{(X_1 + \cdots + X_k - X_{k+1} - \cdots - X_K)}}] \\
    &= \phi_X(t)^k \phi_X(-t)^{K-k}
    = \phi_X(t)^k \overline{\phi_X(t)}^{K-k}.
    \end{aligned}
    \]
    
    Then, since probabilities are nonnegative, we have
    \begin{align*}
        p_k^0 &= |\Pr(S_{k} - S_{K-k}' = 0)| \\
        &= \left| \lim_{T \to \infty} \frac{1}{2 T} \int_{- T}^T \phi_X(t)^k  \overline{\phi_X(t)}^{K-k} dt \right|\\
        &\le  \limsup_{T \to \infty} \frac{1}{2 T} \int_{- T}^T |\phi_X(t)|^K dt \\
        &=\lim_{T \to \infty} \frac{1}{2 T} \int_{- T}^T \phi_X(t)^{K/2}  \overline{\phi_X(t)}^{K/2} dt \\
        &= \Pr(S_{K/2} - S_{K/2}' = 0) \\
        &= p_{K/2}^0,
    \end{align*}
    where the second and fourth identities follow from \Cref{lem:atom-inversion}. Hence the claim follows.

    Returning to the ex-post EF probability, we have
    \begin{align*}
        \Pr(S_{k} \ge S_{K-k}') \cdot \Pr(S_{k} \le S_{K-k}') 
        & = (1-p_k^-)(1-p_k^+) \\
        &= 1-(p_k^- +p_k^+) + p_k^- p_k^+ \\
        &= p_k^0 + p_k^- p_k^+ \\
        &\le p_k^0 + \left(\frac{p_k^- + p_k^+}{2}\right)^2 \\
        &= \left(\frac{1 + p_k^0}{2}\right)^2,
    \end{align*}
    where the inequality follows from the AM-GM inequality.
    The final term is maximized when $k=K/2$, yielding an upper bound of ex-post EF probability.
    Plugging in $k=K/2$, since $p_{K/2}^-=p_{K/2}^+$ by symmetry, equality holds in the AM--GM step. Hence, we have $\Pr(S_{K/2} \ge S_{K/2}') \cdot \Pr(S_{K/2} \le S_{K/2}') = \left(({1 + p_{K/2}^0})/{2}\right)^2$, which matches the upper bound.
    Hence, the ex-post EF probability is maximized when $k=K/2$.
\end{proof}


In the special case where all utilities follow the same distribution $\mathcal{D}$, the above result states that splitting the items equally yields an allocation that is both ex-ante EF and maximizing ex-post EF probability. This may seem natural, but it turns out that the same observation does not generalize for arbitrary $n$.

\begin{restatable}{proposition}{generalnfails}\label{prop:general-n-fails}
    There exists an instance with $n=3$ and for some distribution $\mathcal{D}$, $\mathcal{D}_{ij}=\mathcal{D}$ for all $i,j$, in which an ex-ante EF allocation does not maximize ex-post EF probability.
\end{restatable}
\noindent
The proof in \Cref{sec:genera-n-fails} uses a distribution with both atomic and continuous components.
In fact, we show that for continuous or log-concave generators, two broad classes of distributions, an ex-ante EF allocation still maximizes ex-post EF probability for general $n$.

\begin{theorem}\label{thm:conti-compat}
	Let $\mu$ be a non-atomic distribution. Assume $\mathcal{D}_{1j} = \cdots = \mathcal{D}_{nj} \in \mathcal{C}_\mu$ for all $j$. 
	Then, an ex-ante EF allocation maximizes ex-post EF probability.
\end{theorem}
\begin{proof}
	Again let $X$ be a random variable drawn from $\mu$, and w.l.o.g. assume $\Exp[X] > 0$.
	Recall that $k(M') = \sum_{j \in M'} k_j$ and $K = k(M)$ where $u_i(j) \stackrel{d}{=} S_{k_j}$.
	By the same argument in \Cref{thm:n2homogeneous}, since $\Exp[u_1(S)] = \cdots = \Exp[u_n(S)]$ for any $S \subseteq M$, ex-ante EF allocations are those with $k(A_1) = \cdots = k(A_n) = K/n$.
	Fix an allocation $A$. The ex-post EF probability of $A$ becomes
		\[
		\Pr(u_i(A_i) \ge u_i(A_j), \, \forall i,j \in [n]) = \prod_{i \in N} \Pr(u_i(A_i) = \max_{j \in [n]} u_i(A_j)).
		\]
    We show this probability is bounded above by that of an ex-ante EF allocation.
	If there exists an empty bundle $A_i = \emptyset$, then the ex-post EF probability becomes $0$ since a nonempty bundle $A_j$ has $u_i(A_j) > 0$ a.s. Hence, w.l.o.g. assume each bundle is nonempty.
	
	Let $p_i = \Pr(u_i(A_i) = \max_{j \in [n]} u_i(A_j))$ for all $i \in N$. Then since $u_i(A_j)$ and $u_{i'}(A_j)$ are independent copies for $i \ne i'$, we can fix the agent prefix to $1$: $p_i = \Pr(u_1(A_i) = \max_{j \in [n]} u_1(A_j))$. 
	Then for $i \ne i'$,
		\begin{align*}
			&\Pr\bigl(u_1(A_i) = \max_{j \in N} u_1(A_j),
			u_{1}(A_{i'}) = \max_{j \in N} u_1(A_j)\bigr) \\
            &\qquad \le \Pr(u_1(A_i)=u_1(A_{i'}))=0,
		\end{align*}
	where the identity follows from non-atomicity.
	Hence, we have
	$$\sum_{i \in N}p_i = \sum_{i \in N}\Pr(u_1(A_i) = \max_{j \in [n]} u_1(A_j))=1.$$
	
	Therefore, by the AM-GM inequality, we have \[\prod_{i\in N}p_i \le (\frac{1}{n}\sum_{i \in N}p_i)^n = \frac{1}{n^n},\] with equality holding when $p_1 = \cdots =p_n$. This holds for an ex-ante EF allocation as $u_1(A_1), \ldots, u_1(A_n)$ are identically distributed.
\end{proof}

While \Cref{prop:general-n-fails} uses a distribution which is partly discrete, the broad class of \emph{log-concave} distributions yields positive results for many widely used discrete distributions including the binomial, geometric, and Poisson distributions. 
These distributions have single-peaked pmfs (densities). This is reminiscent of the fact that many problems in social choice behave well in single-peaked domains.

\begin{theorem}\label{thm:logconcave-compat}
	Let $\mu$ be a log-concave distribution. Assume $\mathcal{D}_{1j} = \cdots = \mathcal{D}_{nj} \in \mathcal{C}_\mu$ for all $j$. 
	Then, an ex-ante EF allocation maximizes ex-post EF probability.
\end{theorem}
\noindent%
To prove the theorem, we show a key technical balancing lemma for log-concave distributions, which may be of independent interest. We leave the proof of all three cases in \Cref{sec:pf-technical}.
\begin{restatable}{lemma}{weightbalancing}\label{lem:weight-balancing}
	Let $X$ be a nonnegative discrete log-concave random variable. For an integer $k \ge 0$, let $S_k$ be the sum of $k$ i.i.d. copies of $X$ ($S_0 \equiv 0$) with pmf $p_k$ and cdf $F_k$.
	Let $a,b, x$ be nonnegative integers with $a < b$. Then,
	
	\begin{enumerate}[label={(\arabic*)}]
		\item\label{item:pmf-balancing} $p_a(x)p_b(x)\le p_{a+1}(x)p_{b-1}(x)$,
		\item\label{item:cdf-balancing} $F_a(x)F_b(x)\le F_{a+1}(x)F_{b-1}(x)$,
		\item\label{item:indep-balancing} Let $G: \mathbb{Z}_{\ge 0} \to \mathbb{R}_{\ge 0}$ be a bounded log-concave function. For $k, \ell \ge 0$, let $A_{k,\ell} = \sum_{x \ge 0} p_k(x) F_\ell (x) G(x)$. Then, $A_{a,b} A_{b,a} \le A_{a+1, b-1} A_{b-1,a+1}$.
	\end{enumerate}
\end{restatable}

\noindent
Now we proceed to the proof of the theorem.

\begin{proof}[Proof of \Cref{thm:logconcave-compat}]
	By \Cref{thm:conti-compat}, it is sufficient to consider discrete log-concave distributions. Again let $X$ be a random variable drawn from $\mu$, and w.l.o.g. assume $\Exp[X] > 0$. 
    We write $p_a$ and $F_a$ for the pmf and cdf of $\mu^{*a}$, respectively. 
    Recall that $k(M') = \sum_{j \in M'} k_j$ and $K = k(M)$ where $u_i(j) \sim \mu^{* k_j}$. Ex-ante EF allocations are those with $k(A_1) = \cdots = k(A_n) = K/n$.

    Let $S_{a_1}^{(1)},\ldots,S_{a_n}^{(n)}$ be mutually independent, with $S_{a_i}^{(i)}\sim\mu^{*a_i}$.
    For a nonnegative integral vector $a=(a_1, \ldots, a_n)$ with $\sum_{i=1}^n a_i = K$, let
    \[q_i(a) = \Pr(S_{a_i}^{(i)} \ge S_{a_j}^{(j)}, \, \forall j \ne i).\] Since $\supp(X) \subseteq \mathbb{Z}_{\ge 0}$,
    $$q_i(a) = \sum_{x \ge 0} p_{a_i}(x) \prod_{j \ne i} F_{a_j}(x).$$ 
    Let $\Phi(a) = \prod_{i=1}^n q_i(a)$.
    Given an allocation $A=(A_1, \ldots, A_n)$ with $k(A_i)=a_i$, $q_i(a)$ indicates the probability of agent $i$ not envying anyone, and $\Phi(a)$ becomes the probability of $A$ being EF ex post.
    
    We will show
    that $\Phi$ is maximized at $a = ( \frac{K}{n}, \ldots, \frac{K}{n})$.
    Consider an unbalanced allocation $A$ with $(a_1,\ldots,a_n) \ne (\frac{K}{n}, \ldots,  \frac{K}{n})$. Then there exist $i,j \in N$ such that $a_i < \frac{K}{n} < a_{j}$.
    For the vector $a'$ with $a_i'=a_i + 1$, $a_j' = a_j-1$, and $a_\ell'=a_\ell$ for $\ell\ne i,j$, we claim $\Phi(a') \ge \Phi(a)$.

	First consider $q_k$ for $k \notin \{i,j\}$.
	Let $G_k(x) = \prod_{\ell \notin \{i,j,k\}} F_{a_\ell}(x)$. Then, 
	\begin{align*}
		q_k(a) &= \sum_{x \ge 0}p_{a_k}(x)G_k(x)F_{a_i}(x)F_{a_j}(x),\\
		q_k(a') &= \sum_{x \ge 0}p_{a_k}(x)G_k(x)F_{a_i + 1}(x)F_{a_j - 1}(x).
	\end{align*}
	By \Cref{lem:weight-balancing}~\labelcref{item:cdf-balancing}, we have $q_k(a') \ge q_k(a)$.
	
	Now we compare $q_i$ and $q_j$ together. Let $G(x) = \prod_{\ell \notin \{i,j\}} F_{a_\ell}(x)$. Since each $F_{a_\ell}$ is log-concave (as cdfs of log-concave distributions are log-concave) and takes values between 0 and 1, $G$ is log-concave and bounded. Then defining $A_{k, \ell} = \sum_{x \ge 0} p_{k}(x) F_{\ell} (x) G(x)$, we see that $q_i(a) = A_{a_i, a_j}, q_j(a) = A_{a_j, a_i}$ and $q_i(a') = A_{a_i + 1, a_j -1}, q_j(a') = A_{a_j - 1, a_i + 1}$. Hence, applying \Cref{lem:weight-balancing}~\labelcref{item:indep-balancing}, we have $q_i(a')q_j(a') \ge q_i(a)q_j(a)$.
	
	Therefore,
    \begin{align*}
    \Phi(a')
    &= q_i(a')q_j(a') \prod_{k \notin\{i,j\}}q_k(a')
    \\
    &\ge q_i(a)q_j(a) \prod_{k \notin\{i,j\}}q_k(a) = \Phi(a),
    \end{align*}
	and the claim holds. 

    By repeating the argument until we reach the balanced vector $(a_1,\ldots,a_n) = (\frac{K}{n}, \ldots,  \frac{K}{n})$, we see that the ex-post EF probability is maximized by an ex-ante EF allocation.
    (Note that we can reach the balanced vector in finite steps, as each step decreases $\sum_{\ell=1}^n \left|a_\ell- \frac{K}{n} \right|$ by $2$.)
\end{proof}

Since Poisson distributions also are in the broad family of log-concave distributions, we have the following answer to \Cref{ex:fishing-poisson}, even for $n$ agents.
\begin{restatable}{corollary}{poisson}\label{cor:poisson}
    If $\mathcal{D}_{1j} = \cdots = \mathcal{D}_{nj} = \mathrm{Pois}(\lambda_j)$ for all $j$, then an ex-ante EF allocation maximizes ex-post EF probability.
\end{restatable}
\begin{proof}
    First consider the case where all $\lambda_1, \ldots , \lambda_m$ are rational: $\lambda_j = p_j / q_j$ for some integers $p_j,q_j$. Then each $\lambda_j$ is a multiple of $\lambda = 1/(q_1 \cdots q_m)$, and hence $\mathrm{Pois}(\lambda_j) \in \mathcal{C}_{\mathrm{Pois}(\lambda)}$. Since a Poisson distribution is log-concave, by \Cref{thm:logconcave-compat}, the result holds.

    For real $\lambda_j > 0$, we can approximate each by rational numbers, yielding the result. A detailed proof is in \Cref{sec:poisson}.
\end{proof}

Can we go beyond convolution-power families? It turns out if we consider \emph{symmetric} lottery items for two agents under additional assumptions, an ex-ante EF allocation maximizes ex-post EF probability. We state and prove the result in the following section (\Cref{thm:symmetric-compat}) in a more general sense.


\section{Approximating Ex-Post Fairness}\label{sec:approx}
The positive results in \Cref{sec:EFcompat} assumed that for each fixed item $j$, the distributions $\mathcal{D}_{1j}, \ldots, \mathcal{D}_{nj}$ are identical. However, there are cases where agents have different utility distributions for the same item. In this case, can ex-ante solutions provide at least some approximation to the optimal ex-post probability?

The following example shows that even for the simplest case of \emph{binary lottery items}, where each $u_i(j)$ follow a Bernoulli distribution, we cannot achieve a constant ex-post EF probability independent of $n$.
\begin{example}
	Consider the case of $n$ agents. There are $2n$ items $g_1, \ldots, g_n, g_1', \ldots, g_n'$. For each agent $i$, $u_i(g_i) = 1$ and $u_i(g_i') = 0$ a.s. For $j \ne i$, $u_i(g_j), u_i(g_j') \sim \mathrm{Ber}(1/2)$. Then the allocation $A$ with $A_i = \{g_i, g_i'\}$ is ex-ante EF. However, the ex-post EF probability of $A$ is $(3/4)^{n(n-1)} \to 0$ as $n \to \infty$. 
\end{example}

The ex-post EF probability decays exponentially as the number of agents $n$ increases, at a rate $e^{- \Theta(n^2)}$.
Still, if the utilities are log-concave, we show that an ex-ante EF allocation provides a positive ex-post EF probability depending only on $n$, matching the $e^{- \Theta(n^2)}$ dependence.

We make use of the following Grünbaum-type inequalities.
\begin{lemma}\label{lem:grunbaum}
	Let $X$ be a random variable.
    \begin{enumerate}[label={(\arabic*)}]
    	\item If $X$ is non-atomic log-concave, $\Pr(X \ge \Exp[X]) \ge e^{-1}$.
    	\item If $X$ is discrete log-concave, $\Pr(X > \Exp[X]-1) \ge e^{-1}$.
    \end{enumerate}
\end{lemma}
\begin{proof}
	The non-atomic version was proved by \citet{lovasz2007inequality}. The discrete version was proved by \citet{alqasem2024feige}.
\end{proof}

The following theorem provides an absolute lower bound for the ex-post EF probability under log-concavity. This automatically gives the same approximation ratio to the optimal ex-post EF probability, as probabilities are at most one.

\begin{restatable}{theorem}{logapprox}\label{thm:approx-logconcave-independent} 
	Assume each $\mathcal{D}_{ij}$ is log-concave (all continuous or all discrete).
	Then an ex-ante EF allocation has ex-post EF probability at least $e^{-n(n-1)}$.
\end{restatable}
\begin{proof}
	Fix an allocation $A$. For agents $i\ne j$, let $\Delta_{ij} = u_i(A_i) - u_i(A_j)$. Note that if a random variable $X$ has a log-concave distribution, then so does $-X$. Hence, by \Cref{lem:convolution}, $\Delta_{ij}$ follows a log-concave distribution.
	Now let $A$ be an ex-ante EF allocation. Then, we have $\Exp[\Delta_{ij}] \ge 0$ for every $i,j \in [n]$. By \Cref{lem:grunbaum}, for the case of non-atomic log-concave $\mathcal{D}_{ij}$, $\Pr(\Delta_{ij} \ge 0) \ge e^{-1}$. For the case of discrete log-concave $\mathcal{D}_{ij}$, since each $u_i(j)$ is integer-valued, so is $\Delta_{ij}$. Since $\Exp[\Delta_{ij}] - 1 \ge -1$, $\Pr(\Delta_{ij} \ge 0) \ge \Pr(\Delta_{ij} > \Exp[\Delta_{ij}]-1) \ge e^{-1}$.
	
	Let $q_i = \Pr(u_i(A_i) \ge u_i(A_j), \, \forall j \ne i)$ be the probability agent $i$ does not envy. Then, letting $F_{ij}$ be the cdf of $u_i(A_j)$,
	\begin{align*}
		q_i
		&= \Exp\left[ \Pr \left(u_i(A_i)\ge u_i(A_j), \, \forall j\ne i \,\middle|\, u_i(A_i) \right) \right]\\
		&= \Exp\left[ \prod_{j\ne i} F_{ij} \left(u_i(A_i)\right) \right] \\
        &\ge \prod_{j\ne i}  \Exp\left[ F_{ij} \left(u_i(A_i)\right) \right] \\
		&=  \prod_{j\ne i} \Pr (\Delta_{ij} \ge 0) \\
        &\ge e^{-(n-1)},
	\end{align*}
	where the first inequality follows from applying Chebyshev's association inequality ($\Exp[f(X)g(X)] \ge \Exp[f(X)]\Exp[g(X)]$ for nondecreasing $f$ and $g$) to the nondecreasing functions $F_{ij}$.
	Hence, the ex-post EF probability of $A$ is $\prod_{i=1}^n q_i \ge e^{-n(n-1)}$.
\end{proof}

In \Cref{thm:approx-logconcave-independent}, we have used the fact that calculus on log-concave random variables is well-defined: negation and independent sums preserve log-concavity. Then we were able to make use of \Cref{lem:grunbaum}, which states that the distribution is not too skewed.

Symmetric distributions possess the exact same properties and furthermore it has an even better non-skew\-ness guarantee: $\Pr(X \ge \Exp[X]) \ge \frac{1}{2}$. This gives the following result.

\begin{restatable}{theorem}{approxsym}\label{thm:approx-symmetric-independent} 
	Assume each $\mathcal{D}_{ij}$ is symmetric. Then an ex-ante EF allocation has ex-post EF probability at least $2^{-n(n-1)}$. Furthermore, the bound is tight.
\end{restatable}
\begin{proof}
	The lower bound follows analogously from \Cref{thm:approx-logconcave-independent} and that $q_i \ge 2^{-(n-1)}$.
	For a tight example, consider $n$ items $1, \ldots, n$ with $u_i(i) = 1$ a.s. for all $i$, and $u_i(j) \sim \mathrm{Unif}[0,2]$ for all $j \ne i$. Then $A_i = \{i\}$ is an ex-ante EF allocation. The ex-post EF probability of $A$ is
	\begin{align*}
			\prod_{i=1}^n \Pr(u_i(i) \ge u_i(j), \, \forall j \ne i) 
			= \prod_{i=1}^n \prod_{j \ne i} \Pr(1 \ge u_i(j)) 
			=2^{-n(n-1)}.
	\end{align*}
\end{proof}

These show that different distributions for the same item cannot escape the curse of dimensionality. What if we only relax the convolution-power family assumption of \Cref{sec:EFcompat}?

It turns out that if the common distributions are additionally non-atomic and common for the same item, we can obtain a slightly better approximation ratio to the optimal ex-post EF probability.
\begin{restatable}{proposition}{approxlogconti}\label{thm:0.930approx}
	Assume $\mathcal{D}_{1j} = \cdots = \mathcal{D}_{nj}$ are non-atomic and log-concave for each $j$. Then an ex-ante EF allocation provides an $n^n \left( \frac{e-1}{e^2} \right)^{n \choose 2}$-ap\-prox\-i\-ma\-tion to the optimal ex-post EF probability.
	In particular, for $n=2$, the approximation ratio is $\frac{4(e-1)}{e^2} \approx 0.930$, and this is tight.
\end{restatable}
\begin{proof}
	W.l.o.g. assume that for some item $j$, $\Exp[u_1(j)] > 0$.
	Let $A$ be an ex-ante EF allocation. Let $p_{ij} = \Pr(u_i(A_i) \ge u_i(A_j))$.
	Since $\mathcal{D}_{1j} = \cdots =\mathcal{D}_{nj}$ for each item $j$, we have 
    $\Exp[u_1(A_i)]= \cdots = \Exp[u_n(A_i)]$ for any bundle $A_i$.
    Since $A$ is ex-ante EF, we have 
    $\Exp[u_i(A_1)] = \cdots =\Exp[u_i(A_n)]$ and hence $u_i(A_i) - u_i(A_j)$ is log-concave with zero mean. Hence, by \Cref{lem:grunbaum}, $p_{ij} \ge 1/e$. Similarly, $p_{ji} \ge 1/e$. Since $\mathcal{D}_{ik} = \mathcal{D}_{jk}$ for every item $k$, $p_{ji}=\Pr (u_i(A_j) \ge u_i(A_i))$. Using $\Pr (u_i(A_i) = u_i(A_j)) = 0$ which comes from the fact that the utilities are non-atomic, we have $p_{ij} + p_{ji} = 1$ and hence $1/e \le p_{ij} \le 1- 1/e$. Therefore, $p_{ij}p_{ji} = p_{ij} ( 1- p_{ij}) \ge \frac{1}{e} \left( 1 - \frac{1}{e} \right) = \frac{e-1}{e^2}$.
	
	Let $q_i = \Pr(u_i(A_i) \ge u_i(A_j), \, \forall j \ne i)$ be the probability of agent $i$ not envying.
	As in \Cref{thm:approx-logconcave-independent}, we have $q_i \ge \prod_{j \ne i} p_{ij}$.
	Then, the ex-post EF probability of $A$ is 
	\begin{align*}
		\prod_{i=1}^n q_i \ge \prod_{i=1}^n \prod_{j \ne i} p_{ij}
		= \prod_{1 \le i < j \le n} p_{ij} p_{ji}
		\ge \left(\frac{e-1}{e^2}\right)^{n \choose 2}.
	\end{align*}
	
	On the other hand, consider an arbitrary allocation $A$ with $q_i$ being the probability of agent $i$ not envying.
	By non-atomicity of utilities, as in \Cref{thm:conti-compat}, we have $\sum_{i=1}^n q_i = 1$.
	Then by AM-GM inequality, the ex-post EF probability of $A$ is $\prod_i q_i \le (\frac{1}{n}\sum_{i \in N}q_i)^n = \frac{1}{n^n}$.
	Hence, every ex-ante EF allocation is an $n^n \left( \frac{e-1}{e^2} \right)^{n \choose 2}$-ap\-prox\-i\-ma\-tion to the optimal ex-post EF probability. 
	In particular, for $n=2$, we have an approximation ratio of
	$\frac{4(e-1)}{e^2}$.

    For tightness, see \Cref{sec:0.930}.
\end{proof}

The same holds for symmetric distributions.
\begin{restatable}{proposition}{approxsymconti}\label{thm:symmetric-compat}
	Assume $\mathcal{D}_{1j} = \cdots = \mathcal{D}_{nj}$ are non-atomic and symmetric for each $j$. Then an ex-ante EF allocation provides an $n^n 4 ^{-{n \choose 2}}$-ap\-prox\-i\-ma\-tion to the optimal ex-post EF probability.
	In particular, for $n=2$, an ex-ante EF allocation maximizes ex-post EF probability.
\end{restatable}
\begin{proof}
	Follows analogously from \Cref{thm:0.930approx} and that $p_{ij} = p_{ji} = 1/2$.
\end{proof}

\section{Complexity of Ex-Post Fairness}\label{sec:complexity}
We now study the computational aspects of ex-ante and ex-post envy-freeness.
We restrict our attention to integer-valued random variables with finite support and rational probabilities.
The pmf of each utility $u_i(j)$ is represented explicitly as a list $(x_{\ell},p_{\ell})_{\ell=1,2,\ldots,s}$ where $p_{\ell}=\Pr(u_i(j)=x_{\ell}) \in \mathbb{Q}_{>0}$, $\{x_{\ell}\}_{\ell=1}^s =\supp (u_i(j))$, and $\sum^s_{\ell=1}p_{\ell}=1$. 
The input size is the total bit length of all these support values and rational probabilities.

The ex-ante problems coincide with classical fair division problem after computing the expected utilities $\Exp[u_i(j)]$. This can be done efficiently with our computational model with linearly many arithmetic operations.


Ex-post problems turn out to be computationally hard in a wider range. Indeed, even approximating the optimal ex-post EF probability is NP-hard.
In deterministic instances, the ex-post EF probability of every allocation is either $0$ or $1$. Therefore, any positive multiplicative approximation to the optimal ex-post EF probability would return an EF allocation whenever one exists, and would thus decide the existence of an EF allocation.

Problems that are easy in the ex-ante case turn out to be hard in the ex-post case.
Given an allocation $A$, it is coNP-hard to decide whether this allocation is maximizing ex-post EF probability or not.
Again, if the instance is deterministic and $A$ is not EF, the problem reduces to deciding whether the instance has an EF allocation or not.

This motivates focusing on a more basic computational question: given an allocation, can we compute its ex-post EF probability efficiently?
%
It turns out that the problem is computationally hard.

\begin{restatable}{theorem}{weaksharpP}\label{thm:weaksharpP}
	Even for $n=2$, 
    computing the ex-post EF probability of a given allocation is weakly $\#$P-hard.
\end{restatable}
\begin{proof}
	We perform reduction from \textsc{$\#$Partition}; given positive integers $\{w_1, \ldots, w_m\}$, count the number of partitions $S \subseteq [m]$ that satisfies $\sum_{j \in S}w_j = \sum_{j \in [m] \setminus S}w_j$.
		
	Consider the fair division of lotteries instance with $m+1$ lottery items $1, \ldots, m, {m+1}$. For $j=1, \ldots, m$, the utilities follow the following distribution: $u_i(j) = 2w_j$ with probability $\frac{1}{2}$ and $u_i(j) = 0$ with probability $\frac{1}{2}$, for all agents $i$. Let $u_i({m+1})=W:=\sum_{j=1}^m w_j$ with probability $1$ for all agents $i$. 
    
    Let $D_i = \sum_{j=1}^m u_i(j) - W$ and $p=\Pr(D_1=0)=\Pr(D_2=0)$. 
    Note that
    $$p=\Pr\left(\sum_{j=1}^m u_i(j) = W\right)= \frac{\# \{S \subseteq[m] : \sum_{j \in S}2w_j = W\}}{2^m}.$$
	Since $\sum_{j \in S}2w_j = W$ is equivalent to $\sum_{j \in S}w_j = \sum_{j \in [m] \setminus S}w_j$, we can count the number of partitions, solving \textsc{$\#$Partition}.
    
    Now we compute $p$ by considering the allocation of $A=([m], \{m+1\})$. 
    Note that $\Pr(D_i \ge 0) = \Pr(D_i \le 0)$ as if the realization corresponding to $S \subseteq [m]$, i.e., $u_i(j)=2w_j$ for $j \in S$ and $u_i(j) =  0$ otherwise, yields $D_i = 2\sum_{j \in S}w_j - W \ge 0$, then the realization corresponding to $[m]\setminus S$ yields $D_i = 2\sum_{j \in [m] \setminus S}w_j - W \le 0$. Hence, the probability of the allocation $A=([m], \{m+1\})$ being ex-post EF is 
	$$\Pr\left(D_1 \ge 0\right) \cdot \Pr\left(D_2 \le 0\right)= \left(\frac{1+p}{2}\right)^2,$$
	where we can recover $p$.
\end{proof}

The reduction here came from the weak $\#$P-complete problem of \textsc{$\#$Partition}. 
In fact, even for an arbitrary number of agents, we can evaluate the ex-post EF probabilities in pseudo-polynomial time.

\begin{restatable}{proposition}{DP}\label{prop:evalDP}
	Let $n=|N|$, $m=|M|$, and $B$ be the largest value in the support of $u_i(j)$ among all $i,j$. Then, the ex-post EF probability of an allocation can be computed in $O(n^2m(B+1) + nm^2(B+1)^2)$ arithmetic operations and in pseudo-polynomial time.
\end{restatable}
\begin{proof}
	Fix an allocation $A=(A_1, \ldots, A_n)$. Then $\supp(u_i(A_j)) \subseteq \{0, \ldots, B|A_j| \}$ holds, and we can compute the pmf of $u_i(A_j)$ by dynamic programming. Let $a_j = |A_j|$ and by enumerating $A_j = \{g_{j,1}, \ldots, g_{j, a_j}\}$, define $d(r,s)= \Pr(u_i(g_{j,1})+ \cdots + u_i(g_{j,r})=s)$, while letting $d(0,0)=1, d(0,s) =0$ for $s \ne 0$ and $d(r,s)=0$ for $r<0$.
	By independence, we have 
	$$d(r+1,s) = \sum_{t=0}^B \Pr(u_i(g_{j,r+1})=t)d(r,s-t).$$
	Hence, filling out all $a_j \cdot (Ba_j+1)$ slots of the table, we acquire the pmf of $u_i(A_j)$: $p(x)=d(a_j, x)$. Since each slot takes $O(B+1)$ arithmetic operations to compute, computing the pmf of $u_i(A_j)$ takes $O(a_j^2 (B+1)^2)$ operations. 
    Hence, computing all pmfs for $\{u_i(A_j)\}_{i,j}$ takes $O(nm^2 B^2)$ operations. Then, we can compute the ex-post EF probability
	\begin{align*}
		&\Pr(u_i(A_i) \ge u_i(A_j), \forall i,j \in N) \\ 
		&= \prod_{i \in N} \Pr(u_i(A_i) \ge u_i(A_j), \forall j \in N) \\
		&=  \prod_{i \in N} \sum_{t=0}^{mB}  \Bigl[ \Pr(u_i(A_i) = t)\cdot \prod_{j \ne i} \Pr(u_i(A_j) \le t) \Bigr]
	\end{align*}
	in $O(n^2m(B+1))$ operations, as we can compute the cdfs for all $u_i(A_j)$ beforehand in $O(n^2m(B+1))$ operations. Hence, we have $O(n^2m(B+1)+ nm^2(B+1)^2)$ arithmetic operations in total.
    Since every intermediate probability can be expressed with a denominator dividing the product of the denominators of the input probabilities, its bit length is polynomial in the input size. Hence, the arithmetic-operation bound also yields a pseudo-polynomial running time in the bit model.
\end{proof}

In fact, when we allow utilities to be \emph{correlated}, even only between utilities $u_1(j), \ldots, u_n(j)$ regarding the same item $j$, it turns out we cannot expect a pseudo-polynomial-time algorithm.

\begin{restatable}{theorem}{strongsharpP}\label{thm:strongsharpP}
	If utilities for the same item may be correlated across agents, computing the ex-post EF probability of a given allocation is strongly $\#$P-hard, even with binary lottery items.
\end{restatable}
\noindent
We leave the definitions, proof, and details in \Cref{sec:strongsharpP}.

These hardness results show that ex-post problems are substantially hard to treat in general. 
We cannot generally expect for a certificate for optimal ex-post EF probability unless $\mathrm{NP}=\mathrm{coNP}$, and even computing ex-post EF probability was hard.
However, results in \Cref{sec:EFcompat} found instances where an ex-ante EF allocation maximizes the ex-post EF probability. In such cases, we can ensure optimality for the ex-post problem via the ex-ante problems, even without evaluating the probability.


\section{Identical Valuations}\label{sec:identical}
Our main model assumes that all agent--item utilities are mutually independent. We now consider a different form of uncertainty: every item has the same realized utility for all agents, while values of different items remain independent. Consider the following example.

\begin{example}
    Consider two people dividing 100 one-dollar bills and a lottery ticket that pays 10,000 dollars with probability $0.01$ and zero otherwise.

    One possible allocation gives $100$ dollars to one person and the lottery ticket to the other. Both bundles have expected value $100$, so the allocation is ex-ante EF. However, once the ticket is realized, one person envies the other regardless of the outcome.

    Another option gives each person $50$ dollars and gives the ticket to one of them. The person without the ticket envies the other ex ante. However, with probability $0.99$, the ticket pays zero and both people receive $50$, so the allocation is ex-post EF.
\end{example}

In this example, the ticket has a single realized value shared by both agents. More generally, under \emph{identical valuations}, $u_1(j)=\cdots=u_n(j)$ a.s. for each item $j\in M$. We denote this common utility by $u(j)$ and its distribution by $\mathcal{D}_j$. The random variables $\{u(j)\}_{j\in M}$ remain mutually independent, nonnegative, and of finite mean. Thus, the agents have both the same utility distribution and the same realized utility for each item.

Under identical valuations, an allocation $A$ is ex-post EF if and only if all bundles have the same realized utility. Its ex-post EF probability is therefore
$$\Pr\bigl(u(A_1)=\cdots=u(A_n)\bigr).$$
Similarly, $A$ is ex-ante EF if and only if $\Exp[u(A_1)]=\cdots=\Exp[u(A_n)]$. This observation allows us to extend the convolution-power optimality results and computational guarantees, while also explaining why the positive probability guarantees need not extend.

\subsection{Ex-Ante and Ex-Post Optimality}
The optimality guarantees in \Cref{sec:EFcompat} extend to identical valuations. The counterexample for three agents also has an analogue, given in \Cref{prop:identical-general-n-counterexample}.

\begin{proposition}\label{prop:identical-optimality}
    Under identical valuations, let $\mu$ be a distribution and assume $\mathcal{D}_j\in\mathcal{C}_\mu$ for every item $j$. If $n=2$, or $\mu$ is non-atomic, or $\mu$ is log-concave, then an ex-ante EF allocation maximizes ex-post EF probability.
\end{proposition}
\begin{proof}
    Let $X\sim\mu$. Again, assume w.l.o.g. $\mathbb{E}[X]>0$.
    We use the notation of \Cref{sec:EFcompat}: $S_k$ is a sum of $k$ independent copies of $X$, with $S_0\equiv0$; $u(j)\overset{d}{=}S_{k_j}$; $k(B)=\sum_{j\in B}k_j$; and $K=k(M)$. An allocation is ex-ante EF exactly when $k(A_i)=K/n$ for all $i$. In particular, $K/n$ is an integer.

    For $n=2$, the ex-post EF probability is $\Pr(S_{k(A_1)}=S'_{K-k(A_1)})$, where the two sums are independent. The proof of \Cref{thm:n2homogeneous} shows that this probability is maximized at $k(A_1)=K/2$.

    If $\mu$ is non-atomic, every nonempty bundle has a non-atomic, strictly positive value a.s. Independent nonempty bundle values cannot be equal with positive probability, and an empty bundle cannot equal a nonempty one. Thus every allocation has ex-post EF probability zero, and the conclusion holds vacuously.

    It remains to consider discrete log-concave $\mu$. Let $p_k$ be the pmf of $S_k$. For any nonnegative integer vector $a=(a_1,\ldots,a_n)$ with $\sum_i a_i=K$, define
    $$\Phi(a)=\sum_{x\ge0}\prod_{i\in N}p_{a_i}(x).$$
    For an allocation with $k(A_i)=a_i$, this is its ex-post EF probability, since the bundle values are independent. We maximize $\Phi$ over all such vectors, including those that need not correspond to an allocation.

    If $a\ne(K/n,\ldots,K/n)$, choose $i,j$ with $a_i<K/n<a_j$, and let $a'$ be obtained by increasing $a_i$ by one and decreasing $a_j$ by one. By \Cref{lem:weight-balancing}~\labelcref{item:pmf-balancing},
    $$\Phi(a')-\Phi(a) = \sum_{x\ge0}\bigl(p_{a_i+1}(x)p_{a_j-1}(x)-p_{a_i}(x)p_{a_j}(x)\bigr)
      \prod_{\ell\notin\{i,j\}}p_{a_\ell}(x) \ge 0.$$
    Each step reduces $\sum_i|a_i-K/n|$ by two. Repeating the argument reaches the balanced vector, proving that every ex-ante EF allocation maximizes the probability.
\end{proof}

The Poisson guarantee in \Cref{cor:poisson} also extends to identical valuations. For rational rates, apply \Cref{prop:identical-optimality} with a common Poisson generator. For real rates, use the same rational approximation argument as in the proof of \Cref{cor:poisson}, with equality of the bundle values as the event whose probability is evaluated.

\subsection{Ex-Post Fairness Guarantees}
The positive probability bounds in \Cref{thm:approx-logconcave-independent,thm:approx-symmetric-independent} do not extend to identical valuations, because the agents' non-envy events need not be independent. For example, consider two agents and two items with independent common values $u(1),u(2)\sim\mathrm{Unif}[0,1]$. Giving one item to each agent is ex-ante EF, but its ex-post EF probability is $\Pr(u(1)=u(2))=0$.

More generally, for $n\ge2$ and at least one item, non-atomic common item values imply that every allocation has ex-post EF probability zero. Consequently, the multiplicative guarantees in \Cref{thm:0.930approx,thm:symmetric-compat} hold only vacuously: both the achieved probability and the optimal probability are zero.

\subsection{Computational Complexity}
We use the computational assumptions of \Cref{sec:complexity}: each common utility $u(j)$ is nonnegative and integer-valued, with finite support and rational probabilities given by an explicit pmf. The weak counting hardness and pseudo-polynomial evaluation guarantee both extend to identical valuations.

\begin{proposition}\label{prop:identical-complexity}
    Under identical valuations, computing the ex-post EF probability of a given allocation is weakly $\#$P-hard, even for $n=2$. Let $B$ be the largest value in the support of any $u(j)$. The probability can be computed in $O(nm(B+1)+m^2(B+1)^2)$ arithmetic operations and in pseudo-polynomial time.
\end{proposition}
\begin{proof}
    For hardness, reduce from \textsc{$\#$Partition}. Given positive integers $w_1,\ldots,w_r$, let $W=\sum_{j=1}^r w_j$. Create $r+1$ items with independent common values: for $j\in[r]$, let $u(j)=2w_j$ with probability $1/2$ and $u(j)=0$ otherwise, and let $u(r+1)=W$ a.s. The allocation $A=([r],\{r+1\})$ has ex-post EF probability
    $$\Pr\left(\sum_{j=1}^r u(j)=W\right)
    =\frac{\#\{S\subseteq[r]:\sum_{j\in S}2w_j=W\}}{2^r}.$$
    Multiplying this probability by $2^r$ gives the number of subsets $S$ satisfying $\sum_{j\in S}w_j=\sum_{j\in[r]\setminus S}w_j$, as required.

    For the algorithm, let $b_i=|A_i|$. As in \Cref{prop:evalDP}, dynamic programming computes the pmf of each nonempty bundle value $u(A_i)$ in $O(b_i^2(B+1)^2)$ arithmetic operations. Empty bundles have value zero a.s. Since $\sum_i b_i^2\le m^2$, computing all nonempty-bundle pmfs takes $O(m^2(B+1)^2)$ operations. Independence across bundles then gives
    $$\Pr\bigl(u(A_1)=\cdots=u(A_n)\bigr)
    =\sum_{t=0}^{mB}\prod_{i\in N}\Pr(u(A_i)=t).$$
    Evaluating this expression takes $O(nm(B+1))$ operations, giving the stated total. As in \Cref{prop:evalDP}, the intermediate rational probabilities have bit length polynomial in the input size, so the arithmetic-operation bound yields a pseudo-polynomial running time.
\end{proof}

\section{Conclusion}

We studied the problem of fairly dividing uncertain items. The exogenous uncertainty yielded the ex-ante and ex-post notions of fairness. Focusing on envy-freeness, we were able to identify situations where an ex-ante EF allocation maximizes ex-post EF probability. 
A natural direction for future work will be considering other fairness notions such as EF1, a relaxed notion of EF.

\bibliography{lotteries}

\clearpage
\appendix
\renewcommand{\theHtheorem}{\thesection.\arabic{theorem}}

\section*{Disclosure of AI Usage}
We have used AI models (ChatGPT 6 Astra, 5.6 Sol, 5.5, and 5.4) in the following ways.
First, we have used them to search for related work on the topic. In particular, we asked which works model utilities as random variables in the fair division literature. 
Next, we also used them to search for desirable properties of probability distributions. In particular, the tools helped identify the anti-concentration inequalities of \citet{lovasz2007inequality} and \citet{alqasem2024feige}.
We also used them to generate counterexamples or tight examples. For example, the counterexample for \Cref{prop:general-n-fails} was found by the AI models.
Finally, we used them to check the proofs of the theorems, sometimes yielding an improved result.

\section{Appendix: Proofs Omitted in the Main Body}\label{sec:omitted-proofs}
\subsection{Counterexamples for Three Agents}\label{sec:genera-n-fails}
\generalnfails*
\begin{proof}
    Consider $3000$ items. All agent--item utilities $u_i(j)$ are mutually independent copies of a random variable with distribution
    \[
     \mu=\frac{998}{1000}\delta_1
           +\frac{1}{1000}\operatorname{Unif}[1,2]
           +\frac{1}{1000}\operatorname{Unif}[2000,8000].
    \]
    Call the three mixture components base, low, and high, respectively, and write $p=998/1000$ and $r=1/1000$. Since every item has the same positive expected utility, the ex-ante EF allocations are exactly those with bundle sizes $(1000,1000,1000)$. They all have the same ex-post EF probability.

    We first compute the ex-post EF probability of the balanced, ex-ante EF allocation $A$. 
    Then $u_i(A_j)$ for $1 \le i,j \le 3$ are i.i.d. random variables with the distribution having only one atom at $1000$, and is otherwise non-atomic and greater than $1000$.
    Hence, for agent $1$, the valuations $u_1(A_1), u_1(A_2), u_1(A_3)$ are all equal to $1000$ with probability $R=p^{3000}$ or there exists a unique maximum with probability $1-R$.
    Agent $1$ is envy-free in the former case, and is envy-free with probability $1/3$ in the latter case by symmetry.
    Hence, agent $1$ is envy-free with probability $q_1 = R + \frac{1-R}{3} = \frac{1+2R}{3}$, and $A$ is ex-post envy-free with probability    
    \[
     \Phi(A)=\left(\frac{1+2R}{3}\right)^3
     <\left(\frac{67}{200}\right)^3,
    \]
    using $\frac1{410}<R<\frac1{400}$.

    Now consider an allocation $B$ with bundle sizes $(998,1001,1001)$.
    Then writing $q_i$ for the probability of agent $i$ being envy-free, since $q_2=q_3$ by symmetry, we have ex-post envy-freeness probability $\Phi(B)=q_1q_2^2$. We first bound $q_1$ and then bound $q_2$ in terms of $q_1$.

    We first show that $\frac{5}{16} < q_1 \le \frac{1}{3}$. Trivially, since each item follows the same distribution which realizes to a positive value, we have $q_1 \le \frac{1}{3}$. To show a lower bound, we decompose the utility (to agent $1$) of $1001$-item bundles of agent $2$ into utility of $998$ items $U_2$ and $3$ items $R_2$: $u_1(B_2) = U_2 + R_2$. Similarly for agent $3$: $u_1(B_3) = U_3 + R_3$. Then, $u_1(B_1), U_2, U_3$ are i.i.d. Here, the probability of there existing an item realizing among these three bundles in $[2000,8000]$ is $1-(\frac{999}{1000})^{2994}$. In this event, there exists a unique maximum among $u_1(B_1), U_2, U_3$. By symmetry, the probability of $u_1(B_1)$ being the unique maximum within this event is $\frac{1}{3} \left( 1-(\frac{999}{1000})^{2994} \right)$. Let us call this event $F$.

    We further compute the probability of $u_1(B_1)$ is bigger than $U_2$ and $U_3$ by at least $6$. First consider the case that we get $0< u_1(B_1) -U_2 <6$, within the event $F$.
    That is, $\Pr(0< u_1(B_1) -U_2 <6, F)$. 
    First note that $F$ implies that $B_1$ contains a high-component item. If not, its utility is at most $1996$, whereas either $U_2$ or $U_3$ contains a high-component item which has utility at least $2000$. Then let $j$ be the high item in $B_1$ with the lowest index. Now, without conditioning on $F$, condition on all other utilities than $u_1(j)$, so that $u_1(B_1 \setminus \{j\})$ and $U_2$ are fixed numbers. Then, for $0< u_1(B_1) -U_2 <6$ to occur, $u_1(j)$ must realize in between an interval of length $6$ (which might not be contained in $[2000, 8000]$). Since utilities are independent and $u_1(j)$ is conditioned to be of high type, this probability is  at most $6 \cdot \frac{1}{6000} = \frac{1}{1000}$. Adding $F$ does not increase the probability, yielding the bound.
    The same argument applies to $U_3$. Hence, the probability of $u_1(B_1)$ being bigger than $U_2$ and $U_3$ by at least $6$ is at least $\frac{1}{3} \left( 1- \allowbreak (\frac{999}{1000})^{2994} \right) - \frac{1}{500}$.

    Finally, for $u_1(B_1)$ to be maximum among $u_1(B_1), u_1(B_2), u_1(B_3)$, it suffices that all items in $R_2$ and $R_3$ realize to at most $2$, yielding $R_2 \le 6, R_3 \le 6$. This is independent the previous event and has probability $\left( \frac{999}{1000}\right)^6$.
    Hence, we obtain the lower bound of
    $$q_1 \ge \left( \frac{999}{1000}\right)^6 \left[\frac{1}{3} \left( 1-\left(\frac{999}{1000}\right)^{2994} \right) - \frac{1}{500} \right]>\frac{5}{16}.$$

    Now we bound $q_2$ in terms of $q_1$. 
    Note that since $u_1(B_j), u_2(B_j),\allowbreak u_3(B_j)$ are i.i.d., we can express $q_2$ in terms of agent $1$'s valuation: $q_2 = \Pr(u_1(B_2) \ge u_1(B_1), u_1(B_2) \ge u_1(B_3) )$. 
    Then let $Z_i$ be the indicator random variable on whether $u_1(B_i) = \max_{j\in[3]} u_1(B_j)$ and let $Z$ be the number of bundles that attains the maximum among $u_1(B_1),u_1(B_2),u_1(B_3)$. Then, we have $\Exp[Z] = \Exp[Z_1] + \Exp[Z_2] + \Exp[Z_3]=q_1 + 2q_2$. Since $Z \ge 1$ a.s., we have $\Exp[Z] \ge 1 + \Pr(Z \ge 2)$. For the event $\{Z \ge 2\}$, it includes the event $E$ where all $1001$ items in $B_2$ and all $1001$ items in $B_3$ realize to $1$ for agent $1$, and at most $3$ items in $B_1$ realize to a value in $(1,2]$ while others realize to $1$.
    Then, we have 
    \begin{align*}
        \Pr(Z \ge 2) &\ge \Pr(E) \\
        &= p^{2002} \sum_{j=0}^3 \binom{998}{j} \left(\frac{1}{1000}\right)^j p^{998-j} \\
        &= p^{3000} \sum_{j=0}^3 \binom{998}{j} \left(\frac{1}{1000p}\right)^j \\
        &> \frac{1}{160}.
    \end{align*}

    Thus, we have $q_1 + 2q_2 \ge 1 + \Pr(Z \ge 2) > 1+\frac{1}{160}$, or in other words,
    $$q_2 > \frac{1+\frac{1}{160}-q_1}{2}.$$
    
    Hence, we have $\Phi(B) > q_1 \left(\frac{1+\frac{1}{160}-q_1}{2} \right)^2$. Since the right hand side is increasing in $0 \le q_1 \le 1/3$ and $q_1 > \frac{5}{16}$ as we attained before, we have
    $$\Phi(B) > \frac{5}{16} \left(\frac{1+\frac{1}{160}-\frac{5}{16}}{2} \right)^2 = \frac{5}{16} \left(\frac{111}{320}\right)^2 > \left(\frac{67}{200}\right)^3 > \Phi(A).$$

    Thus, no ex-ante EF allocation maximizes the ex-post EF probability.
\end{proof}

There is also a counterexample for identical valuations.

\begin{proposition}\label{prop:identical-general-n-counterexample}
    There exists an instance with $n=3$ and identical valuations such that the common item values $u(j)$ are i.i.d. and an ex-ante EF allocation does not maximize ex-post EF probability under identical valuations.
\end{proposition}
\begin{proof}    
Consider six items whose utilities all follow the following distribution:
    \[
     \mu=\frac{3}{4}\operatorname{Unif}\{1,3,5, \ldots, 29\} + \frac{1}{4}\delta_{60}.
    \]
Since all items' utilities follow the same distribution, an allocation is ex-ante EF if and only if its bundle sizes are $(2,2,2)$. Under identical valuations, ex-post EF is equivalent to equality of the three bundle values.

Let $L=\{1,3,\ldots,29\}$ and let $c_s$ count the ordered pairs in $L^2$ with sum $s$. The nonzero $c_s$'s form the sequence $1,2,\ldots,14,15,14,\ldots,2,1$. Hence
\[
  \sum_s c_s^2=15^2+2\sum_{j=1}^{14}j^2=2255,
  \quad
  \sum_s c_s^3=15^3+2\sum_{j=1}^{14}j^3=25425.
\]
A two-item bundle has value $s\in L+L$ with probability $c_s/400$, each value in $60+L$ with probability $10/400$, and value $120$ with probability $25/400$. These three sets are disjoint. Therefore an allocation $A$ with bundle sizes $(2,2,2)$ has EF probability
\[
  \Pr(A\text{ is EF})
  =\frac{25425+15\cdot10^3+25^3}{20^6}
  =\frac{2242}{2560000}.
\]
Now consider an allocation $B$ with bundle sizes $(1,1,4)$. It realizes to be EF only when both singleton values are $60$ and the four remaining values belong to $L$ and sum to $60$: four elements of $L$ have even sum, whereas any four-item bundle containing $60$ has value greater than $60$.
By symmetry, $c_s=c_{60-s}$, so there are $\sum_s c_sc_{60-s}=2255$ such ordered quadruples. Thus
\[
  \Pr(B\text{ is EF})
  =\left(\frac14\right)^2\frac{2255}{20^4}
  =\frac{2255}{2560000}
  >\Pr(A\text{ is EF}).
\]
Consequently, no ex-ante EF allocation is ex-post EF probability maximizing.
\end{proof}

\subsection{Proof of \Cref{lem:weight-balancing}}\label{sec:pf-technical}

\weightbalancing*

We first introduce some tools we need: stochastic orders. For two discrete random variables $X$ and $Y$, we say \emph{$Y$ (first-order) stochastically dominates $X$} if $\Pr(X>x) \le \Pr(Y>x)$ for all $x$ and write $X \leST Y$.
We say \emph{$X$ is smaller than $Y$ in the likelihood ratio order} if $\Pr(X=x)\Pr(Y=y) \ge \Pr(X=y)\Pr(Y=x)$ for all $x \le y$ and write $X \leLR Y$. The following implication holds.

\begin{lemma}\label{lem:lr-implies-st}
	Let $X$ and $Y$ be discrete random variables. If $ X \leLR Y$, then $X \leST Y$.
\end{lemma}
\begin{proof}
	See, e.g., \citet{ShakedShanthikumar2007}.
\end{proof}

Now we proceed to prove \labelcref{item:pmf-balancing,item:cdf-balancing}.
\begin{proof}[Proof of \Cref{lem:weight-balancing}~\labelcref{item:pmf-balancing,item:cdf-balancing}]
	The inequalities are trivial for $b = a+1$. We assume $b > a + 1$.
	
	\paragraph{Proof of \labelcref{item:pmf-balancing}.}
	We first claim that for all integer $a\ge 0$ and integers $x\ge y$,
	$$p_{a+1}(x)p_a(y)\ge p_a(x)p_{a+1}(y).$$
	For $a=0$, this is immediate from $p_0=\delta_0$, the Dirac measure concentrated at $0$. The right-hand side is positive only if $x=y=0$ (as $X \ge 0$), which then matches the left-hand side. 
	Now assume $a \ge 1$. Since $X \ge 0$, we have $p_{a+1}(x)=\sum_{u\ge 0} p_a(x-u)p_1(u)$.
	Therefore,
		\begin{align*}
		&p_{a+1}(x)p_a(y)-p_a(x)p_{a+1}(y)\\
        &\quad = \sum_{u\ge 0}p_1(u)
		\left( p_a(x-u)p_a(y)-p_a(x)p_a(y-u) \right).
		\end{align*}
	
	Since $u \ge 0$ and $x \ge y$, we have either $y-u \le x-u \le y \le x$ or $y-u \le y \le x-u  \le x$. In either case, as $p_a$ is log-concave by \Cref{lem:convolution}, we have $p_a(x-u)p_a(y) \ge p_a(x)p_a(y-u)$.
	Hence, the claim follows as each summand is nonnegative.
	
	Next, we show that our inequality holds.
	Since $a<b$, writing $d = b-a-1 \ge 0$, we have
	\begin{align*}
		&p_{a+1}(x)p_{b-1}(x)-p_a(x)p_b(x) \\ 
		&=
			p_{a+1}(x)
			\left(\sum_{t\ge 0}p_a(x-t)p_d(t) \right) \\
            &\quad -
			p_a(x)
			\left(\sum_{t\ge 0}p_{a+1}(x-t)p_d(t)\right) \\
		&=
		\sum_{t\ge 0}p_d(t)
		\left(
		p_{a+1}(x)p_a(x-t)-p_a(x)p_{a+1}(x-t)
		\right).
	\end{align*}
	
	Since \(t\ge 0\), we have \(x\ge x-t\). Hence by the previous claim, each summand is nonnegative, proving the desired inequality.
	
	\paragraph{Proof of \labelcref{item:cdf-balancing}.}
    We claim ${F_a(x)}{F_{a+2}(x)} \le {F_{a+1}(x)^2}$ holds for every $a\ge 0$ and $x \ge 0$.
    Since the inequality is trivial if $F_{a}(x)=0$, assume $F_a(x) > 0$. If $F_{a+1}(x)=0$, then also $F_{a+2}(x)=0$ holds since $S_{a+2} \ge S_{a+1}$. Hence the inequality is again trivial. Assume $F_{a+1}(x) > 0$.
    
    Define the conditional distribution $\nu_a (y) = \Pr(S_a = y | S_a \le x) = \frac{p_a(y)}{F_a(x)} \mathbf{1}(y \le x)$ (similarly, $\nu_{a+1}$).
	Then, letting $R_x(y) = \Pr(X \le x - y)$, we have
    \begin{align*}
        \frac{F_{a+1}(x)}{F_a(x)} &= \frac{\Pr(S_a + X \le x)}{\Pr(S_a \le x)} \\
        &= \frac{\sum_{y \le x} \Pr(S_a =y) \Pr(X \le x-y)}{F_a (x)} \\
                &= \sum_{y \le x}\frac{p_a(y) R_x(y)}{F_a (x)} \\
        &= \sum_{y} \nu_a(y) R_x(y)
        \\&= \Exp_{Y \sim \nu_a}[R_x(Y)].
    \end{align*}
	We then observe that $\nu_{a}$ is smaller than $\nu_{a+1}$ in the likelihood ratio order, i.e. $\nu_a \leLR \nu_{a+1}$.
    First note that, for any $y\le z\le x$, as in \labelcref{item:pmf-balancing},
    \[
        p_a(y)p_{a+1}(z)
        \ge
        p_a(z)p_{a+1}(y).
    \]
    Therefore,
    \begin{align*}
        \nu_a(y)\nu_{a+1}(z)
        =
        \frac{p_a(y)p_{a+1}(z)}
             {F_a(x)F_{a+1}(x)}
        \ge
        \frac{p_a(z)p_{a+1}(y)}
             {F_a(x)F_{a+1}(x)}
        =
        \nu_a(z)\nu_{a+1}(y).
    \end{align*}
    Thus, $\nu_a\le_{\mathrm{lr}}\nu_{a+1}$, and by
    \Cref{lem:lr-implies-st}, $\nu_a\le_{\mathrm{st}}\nu_{a+1}$.
	
	Since $R_x(y)$ is nonincreasing in $y$, stochastic dominance yields $ \Exp_{Y \sim \nu_{a+1}}[R_x(Y)] \le  \Exp_{Y \sim \nu_a}[R_x(Y)]$, i.e., $\frac{F_{a+2}(x)}{F_{a+1}(x)} \le \frac{F_{a+1}(x)}{F_a(x)}$.
	Iteratively applying the inequality with $a, a+1, \ldots, b-2$ implies the desired inequality.	
\end{proof}

To prove \labelcref{item:indep-balancing}, we need further tools.
The following lemma generalizes the fact that cdf of a log-concave random variable is log-concave.
\begin{lemma}\label{lem:cdf-tail-logconcave}
	Let $f:\mathbb{Z} \to \mathbb{R}_{\ge 0}$ be a summable log-concave function. Then, $\overline F(x)=\sum_{y\ge x} f(y)$ is log-concave.
\end{lemma}
\begin{proof}
The claim is immediate if \(f\equiv0\). Otherwise, let $Z=\sum_{y\in\mathbb Z}f(y)>0$ and define $p(x)={f(x)}/{Z}$.
Then \(p\) is a log-concave pmf. Let \(Y\) have pmf \(p\). Then, \(-Y\) is log-concave, and
\[
    \overline F(x)
    =
    Z\Pr(Y\ge x)
    =
    Z\Pr(-Y\le -x).
\]
Therefore, \(\overline F\) is log-concave as the cdf of a log-concave random variable is log-concave.
\end{proof}

A nonnegative function (or an infinite matrix) $K:\mathbb{Z}_{\ge 0}^2 \to \mathbb{R}_{\ge 0}$ is called \emph{totally positive of order 2} ($\mathrm{TP}_2$) if for all $i \le i'$ and $j \le j'$,
$$K(i,j)K(i',j') \ge K(i,j')K(i',j).$$

A nonnegative function $f : \mathbb{Z}_{\ge 0} \to \mathbb{R}_{\ge 0}$ is called a \emph{Pólya frequency function of order 2} ($\mathrm{PF}_2$) if $T:\mathbb{Z}_{\ge 0}^2 \to \mathbb{R}_{\ge 0}$ defined by $T(i,j) = f(j-i)$ for $i \le j$ and $T(i,j) = 0$ for $i > j$ is $\mathrm{TP}_2$. The function $T$ is called the \emph{Toeplitz function}.
Total positivity is closely related to log-concave functions.
The following lemmas are classical results in total positivity.
\begin{lemma}[{\citet{schoenberg1951polya}}]\label{lem:pf2-toeplitz}
	A nonnegative function $f : \mathbb{Z}_{\ge 0} \allowbreak\to \mathbb{R}_{\ge 0}$ is log-concave if and only if it is $\mathrm{PF}_2$.
\end{lemma}

\begin{lemma}[\citet{Karlin1968}]\label{lem:tp2-composition}
	Let $K(i,x)$ and $L(x,j)$ be nonnegative $\mathrm{TP}_2$ functions. If the \emph{matrix product} $M$ defined by
	$$ M(i,j)=\sum_{x \ge 0} K(i,x)L(x,j)$$
	has finite entries, then $M$ is $\mathrm{TP}_2$.
\end{lemma}

Based on these, we first show that $p_k(x)$ is $\mathrm{TP}_2$ when considered as a two-variable function of $(k,x)$.
\begin{lemma}\label{lem:pkTP2}
	Let $X$ be a nonnegative discrete log-concave random variable. For an integer $k \ge 0$, let $S_k$ be the sum of $k$ i.i.d. copies of $X$ ($S_0 \equiv 0$) with pmf $p_k$.
	Then, $p_k(x)$ is $\mathrm{TP}_2$ in $(k,x)$.
\end{lemma}
\begin{proof}
	Since $p_1$ is log-concave, the Toeplitz function $T(x,y)=p_1(y-x)$ is $\mathrm{TP}_2$ by \Cref{lem:pf2-toeplitz}.
	We claim that $p_k \leLR p_{k+1}$ holds for all $k$. For $k=0$, this is trivial as $X \ge 0$. Assume $p_{k-1} \leLR p_{k}$. Then by definition, the two-row matrix $K$ with the first row being $p_{k-1}$ and the second row being $p_k$ is $\mathrm{TP}_2$ (we regard all other rows as zero rows). By \Cref{lem:tp2-composition}, the product matrix $KT$ is $\mathrm{TP}_2$. Since
		\begin{align*}
		(KT)(0,y)
		= \sum_{x \ge 0} K(0,x) T(x,y) 
		= \sum_{x \ge 0} p_{k-1}(x) p_1(y-x) 
		= p_k(y),
		\end{align*}
	the first row of $KT$ is $p_k$. Similarly, the second row of $KT$ is $p_{k+1}$ and the rest rows are identically zero. Since $KT$ is $\mathrm{TP}_2$, we have $p_k \leLR p_{k+1}$.
	
	Then, by transitivity, for $k \le k'$, we have $p_k \leLR p_{k'}$. This by definition is for all $k \le k'$ and $x \le x'$, $p_k(x)p_{k'}(x') \ge p_k(x')p_{k'}(x)$. Hence, $p_k(x)$ is $\mathrm{TP}_2$ in $(k,x)$.
\end{proof}

Using this, we prove that a weighted expectation of $S_k$ becomes log-concave in $k$.
\begin{lemma}\label{lem:weighted-expectation}
	Let $X$ be a nonnegative discrete log-concave random variable. For an integer $k \ge 0$, let $S_k$ be the sum of $k$ i.i.d. copies of $X$ ($S_0 \equiv 0$) with pmf $p_k$ and cdf $F_k$.
	Let $W:\mathbb{Z}_{\ge 0}\to\mathbb{R}_{\ge 0}$ be a bounded log-concave function.
	Define $B_k=\Exp[W(S_k)]=\sum_{x \ge 0} p_k(x)W(x)$.
    Then, \(k\mapsto B_k\) is log-concave. In particular, $B_{k+1}^2\ge B_{k}B_{k+2}$.
\end{lemma}
\begin{proof}
	The proof proceeds analogously to that of \Cref{lem:weight-balancing}~\labelcref{item:cdf-balancing}.

    The claim is immediate if \(W\equiv 0\). We first prove the claim under the additional assumption that \(W(x)>0\) for every $x\in\mathbb Z_{\ge 0}$. The general case will follow by approximation.
	
	First observe that
	\begin{align*}
		B_{k+1} &= \sum_{y \ge 0} p_{k+1}(y)W(y) \\
		&= \sum_{y \ge 0} \left( \sum_{x \ge 0 }p_{k}(x)p_1(y-x) \right) W(y) \\
		&= \sum_{x \ge 0} p_k(x) \sum_{y \ge 0} p_1(y-x) W(y).
	\end{align*}
	
	Define $\mathcal{T}W(x) = \sum_{y \ge 0} p_1(y-x) W(y) =  \sum_{t} p_1(t) W(x+t)$ and $R(x) = {\mathcal{T}W(x)}/{W(x)}$. We claim that $R(x)$ is nonincreasing in $x$.
	
	First observe that since $W$ is positive and log-concave, $W(s+1)/W(s)$ is nonincreasing in $s$. Hence, for any $t \ge 0$,
	$$\frac{W(x+t)}{W(x)} = \frac{W(x+1)}{W(x)} \frac{W(x+2)}{W(x+1)} \cdots \frac{W(x+t)}{W(x+t-1)}$$
	is nonincreasing in $x$. Then, \[R(x) = {\mathcal{T}W(x)}/{W(x)} = \sum_{t} p_1(t) \frac{W(x+t)}{W(x)}\] is nonincreasing in $x$ as each summand is.
	
	For $B_k > 0$, define the probability distribution $\nu_k(x) = \frac{p_k(x)W(x)}{B_k}$. Then,
		\begin{align*}
		\frac{B_{k+1}}{B_k}
		&= \frac{1}{B_k} \sum_{x \ge 0}p_k(x) \mathcal{T}W(x) \\
		&=   \sum_{x \ge 0} \frac{p_k(x) W(x)}{B_k} \frac{\mathcal{T}W(x)}{W(x)} \\
		&= \Exp_{Z \sim \nu_k} [R(Z)] .
		\end{align*}
	
	By \Cref{lem:pkTP2}, $p_k(x)$ is TP2 in $(k,x)$. Hence, for $x \le y$, \[p_k(x)p_{k+1}(y) \ge p_k(y) p_{k+1}(x).\]
	Multiplying the nonnegative quantity $\frac{W(x)W(y)}{B_kB_{k+1}}$ on both sides yields
	$\nu_k(x)\nu_{k+1}(y) \ge \nu_k(y)\nu_{k+1}(x)$, i.e., $\nu_k \leLR \nu_{k+1}$.
	By \Cref{lem:lr-implies-st}, $\nu_k \leST \nu_{k+1}$, and since $R$ is nonincreasing, we have
	$\Exp_{Z \sim \nu_k}[R(Z)] \ge \Exp_{Z \sim \nu_{k+1}}[R(Z)]$. That is,
	$$\frac{B_{k+1}}{B_k} \ge \frac{B_{k+2}}{B_{k+1}}.$$

    We now remove the strict-positivity assumption on $W$. Extend \(W\) to
    \(\mathbb Z\) by setting \(W(x)=0\) for \(x<0\). For
    \(\epsilon\in(0,1)\), define the probability mass function
    \[
        q_\epsilon(t)
        =
        \frac{1-\epsilon}{1+\epsilon}
        \epsilon^{|t|},
        \qquad t\in\mathbb Z,
    \]
    and let
    \[
        W_\epsilon
        =
        W*q_\epsilon.
    \]
    The function \(q_\epsilon\) is strictly positive and
    log-concave. Hence, by closure under convolution,
    \(W_\epsilon\) is strictly positive (as there exists $x$ such that $W(x) > 0$) and log-concave. 
    Moreover, since \(W\) is bounded and \(q_\epsilon\) is a probability mass function, $0<W_\epsilon(x)\le \|W\|_\infty$ holds for every $x$.
    
    We further claim that \(W_\epsilon \to W\) uniformly as $\epsilon \to 0$. Since $ q_\epsilon(0) = \frac{1-\epsilon}{1+\epsilon} \to 1$ and
    $\sum_{t\ne0}q_\epsilon(t) = \frac{2\epsilon}{1+\epsilon} \to 0$,
    we have for any $x$,
    \begin{align*}
        |W_\epsilon(x)-W(x)|
        &\le \sum_{t\ne0}q_\epsilon(t)|W(x-t)-W(x)| \\
        &\le \|W\|_\infty\sum_{t\ne0}q_\epsilon(t)
         = \frac{2\epsilon}{1+\epsilon}\|W\|_\infty \to 0.
    \end{align*}
    Hence, $W_\epsilon\to W$ 
    uniformly as
    \(\epsilon\to0\).

    Define
    \[
        B_k^{(\epsilon)}
        =
        \sum_{x\ge0}p_k(x)W_\epsilon(x).
    \]
    Applying the strictly positive case to \(W_\epsilon\) gives
    \[
        \left(B_{k+1}^{(\epsilon)}\right)^2
        \ge
        B_k^{(\epsilon)}B_{k+2}^{(\epsilon)}.
    \]
    Moreover, for every fixed \(r\ge0\),
    \begin{align*}
        \left|B_r^{(\epsilon)}-B_r\right|
        &= \left|\sum_{x\ge0}p_r(x)
           \bigl(W_\epsilon(x)-W(x)\bigr)\right| \\
        &\le \sum_{x\ge0}p_r(x)\left|W_\epsilon(x)-W(x)\right|
         \le \sup_{x\ge0}\left|W_\epsilon(x)-W(x)\right|.
    \end{align*}
    Since \(W_\epsilon\to W\) uniformly,
    $B_r^{(\epsilon)}\to B_r$ as $\epsilon \to 0$.
    Letting \(\epsilon\to0\) therefore yields
    \[
        B_{k+1}^2
        \ge
        B_kB_{k+2},
    \]
    as required.

    It remains to verify that \(k\mapsto B_k\) has integer-interval support. Assume \(W\not\equiv0\), and write
    \[
        \supp(X)=[\alpha,\beta]\cap\mathbb Z_{\ge0},
        \quad
        \supp(W)=[L,U]\cap\mathbb Z_{\ge0},
    \]
    allowing \(\beta,U=\infty\). For \(k\ge1\),
    \[
        \supp(S_k)=[k\alpha,k\beta]\cap\mathbb Z_{\ge0}.
    \]
    Since all summands defining \(B_k\) are nonnegative,
    \[
        B_k>0
        \iff \supp(S_k)\cap\supp(W)\ne\varnothing
        \iff k\alpha\le U \text{ and } k\beta\ge L.
    \]
    The integers satisfying these inequalities form an interval.
    Moreover, \(B_0=W(0)>0\) only if \(L=0\), in which case adjoining
    \(k=0\) preserves the interval property. Thus,
    \(k\mapsto B_k\) is log-concave.
\end{proof}

Now we proceed to prove \Cref{lem:weight-balancing}~\labelcref{item:indep-balancing}.
\begin{proof}[Proof of \Cref{lem:weight-balancing}~\labelcref{item:indep-balancing}]
	We first claim that \[A_{k,\ell} = \sum_{x \ge 0} p_k (x) F_\ell(x)G(x)\] is $\mathrm{TP}_2$ in $(k, \ell)$.
	To show this, we first show that $F_\ell (x)$ is $\mathrm{TP}_2$ in $(\ell, x)$. Note that for the function $J(y,x) = \mathbf{1}(y \le x)$, $F_\ell(x) = \sum_{y \ge 0} p_\ell (y) J(y,x)$. We observe that $J(y,x)$ is $\mathrm{TP}_2$. For $x \le x'$ and $y \le y'$, if $J(y',x) = 1$, we have $y \le y' \le x \le x'$. This implies $J(y,x) = J(y',x')=1$. Hence, $J(y,x)$ is $\mathrm{TP}_2$. Since $p_\ell(x)$ is also $\mathrm{TP}_2$ in $(\ell, x)$ (\Cref{lem:pkTP2}), by \Cref{lem:tp2-composition}, $F_\ell(x)$ is $\mathrm{TP}_2$ in $(\ell, x)$. Transposition preserves $\mathrm{TP}_2$, so $F_\ell(x)$ is also $\mathrm{TP}_2$ in $(x,\ell)$.
	Note that since $G(x) \ge 0$, $F_\ell(x)G(x)$ is also $\mathrm{TP}_2$ in $(x,\ell)$.
	Hence, $A_{k,\ell} = \sum_{x \ge 0} p_k (x) \left( F_\ell(x)G(x) \right)$ is $\mathrm{TP}_2$ in $(k, \ell)$ again by \Cref{lem:tp2-composition}.
	
	Then, we claim that $A_{k,\ell}$ is log-concave in each component $k$ and $\ell$. First fix $\ell$. Let $W_\ell (x) = F_\ell(x) G(x)$. It is a log-concave function as both $F_\ell$ and $G$ are. It is also bounded, since $0\le F_\ell(x)\le 1$ and \(G\) is bounded. Hence by \Cref{lem:weighted-expectation}, $A_{k,\ell} = \sum_{x \ge 0} p_k (x) W_\ell (x)$ is log-concave in $k$. Next, fix $k$ and let $H_k(y) = \sum_{x \ge y} p_k(x) G(x)$.
	\begin{align*}
		A_{k,\ell} &= \sum_{x \ge 0} p_k (x)  F_\ell(x)G(x) \\
		&= \sum_{x \ge 0} p_k (x) G(x) \sum_{y \le x} p_\ell(y) \\
		&= \sum_{y \ge 0} p_\ell(y) \sum_{x \ge y} p_k(x) G(x) \\
		&= \sum_{y \ge 0} p_\ell(y) H_k(y).
	\end{align*}
	Since $p_k(x)$ and $G(x)$ are log-concave in $x$, so is $p_k(x)G(x)$. Hence, by \Cref{lem:cdf-tail-logconcave}, $H_k(y)$ is log-concave in $y$. 
    Moreover,
    \[
        0\le H_k(y)
        \le
        \|G\|_\infty\sum_{x\ge y}p_k(x)
        \le
        \|G\|_\infty,
    \]
    so $H_k$ is bounded.
    Hence, by \Cref{lem:weighted-expectation}, $A_{k,\ell}$ is log-concave in $\ell$.

    We now use $\mathrm{TP}_2$ and componentwise log-concavity to prove that, for $r,c\ge1$,
    \[
        A_{r,c}^2
        \ge
        A_{r-1,c+1}A_{r+1,c-1}.
    \]
    
    First suppose that \(A_{r,c}=0\). We claim that
    \[
        A_{r-1,c+1}A_{r+1,c-1}=0.
    \]
    Suppose otherwise that both factors are positive. Since
    \(A_{k,\ell}\) is \(\mathrm{TP}_2\), applying the
    \(\mathrm{TP}_2\) inequality to rows \(r-1,r+1\) and columns
    \(c-1,c+1\) gives
    \[
        A_{r-1,c-1}A_{r+1,c+1}
        \ge
        A_{r-1,c+1}A_{r+1,c-1}
        >0.
    \]
    Hence, $A_{r-1,c-1}>0$ and $A_{r+1,c+1}>0$ hold.
    Together with $A_{r-1,c+1}>0$ and $A_{r+1,c-1}>0$, log-concavity of $A_{r,c}$ in $c$ yields
    $A_{r-1,c}>0$ and $A_{r+1,c}>0$.
    Log-concavity of $A_{r,c}$ in $r$ then gives
    \[
        A_{r,c}^2
        \ge
        A_{r-1,c}A_{r+1,c}
        >0,
    \]
    contradicting \(A_{r,c}=0\). Thus, the desired inequality holds
    when \(A_{r,c}=0\).
    
    Now assume \(A_{r,c}>0\). By \(\mathrm{TP}_2\),
    \[
        A_{r-1,c}A_{r,c+1}
        \ge
        A_{r-1,c+1}A_{r,c}
    \]
    and
    \[
        A_{r,c-1}A_{r+1,c}
        \ge
        A_{r,c}A_{r+1,c-1}.
    \]
    Multiplying the two inequalities gives
    \[
        A_{r-1,c}A_{r+1,c}
         A_{r,c-1}A_{r,c+1}\ge
        A_{r-1,c+1}A_{r+1,c-1}A_{r,c}^2.
    \]
    Componentwise log-concavity gives
    \[
        A_{r,c}^2
        \ge
        A_{r-1,c}A_{r+1,c}
    \]
    and
    \[
        A_{r,c}^2
        \ge
        A_{r,c-1}A_{r,c+1}.
    \]
    Therefore,
    \[
        A_{r,c}^4
        \ge A_{r-1,c}A_{r+1,c}A_{r,c-1}A_{r,c+1}
        \ge A_{r-1,c+1}A_{r+1,c-1}A_{r,c}^2.
    \]
    Since \(A_{r,c}>0\), division by \(A_{r,c}^2\) gives
    \[
        A_{r,c}^2
        \ge
        A_{r-1,c+1}A_{r+1,c-1},
    \]
    and hence the claimed inequality holds.

    Finally, given $a < b$, let $B_r = A_{r,a+b-r}$. Then, $B_r^2 = A_{r,a+b-r}^2 \ge A_{r-1, a+b-r+1} A_{r+1,a+b-r-1} = B_{r-1} B_{r+1}$.
    Then
    \[
        B_r^2\ge B_{r-1}B_{r+1}
        \qquad (a<r<b).
    \]
    
    If \(B_aB_b=A_{a,b}A_{b,a}=0\), the desired inequality is immediate.
    Otherwise, $A_{a,b}>0$ and $A_{b,a}>0$ hold, and since $A_{k, \ell}$ is \(\mathrm{TP}_2\),
    \[
        A_{a,a}A_{b,b}
        \ge
        A_{a,b}A_{b,a}
        >0.
    \]
    Since \(A_{k,\ell}\) is log-concave in each coordinate, every
    row and column has interval support. Hence,
    \[
        A_{r,\ell}>0
        \qquad
        \text{for all }r,\ell\in\{a,\ldots,b\}.
    \]
    In particular, \(B_r>0\) for every \(r=a,\ldots,b\).
    
    Therefore, the ratios \(B_{r+1}/B_r\) are nonincreasing, and
    \[
        \frac{B_{a+1}}{B_a}
        \ge
        \frac{B_b}{B_{b-1}}.
    \]
    Thus,
    \[
        B_aB_b\le B_{a+1}B_{b-1},
    \]
    which is equivalent to
    \[
        A_{a,b}A_{b,a}
        \le
        A_{a+1,b-1}A_{b-1,a+1}.
    \]
    
\end{proof}

\subsection{Proof of \Cref{cor:poisson}}\label{sec:poisson}
\poisson*
\begin{proof}
    It remains to prove the case for real $\lambda_j >0$. Fix an ex-ante EF allocation $A$. Then, for each $\lambda_j$, we construct a rational sequence $(\lambda_j ^{(k)} )_k$ converging to $\lambda_j$ such that for each $k$, the allocation $A$ is ex-ante EF for the stage-$k$ approximation instance where the common distribution for item $j$ is $\mathrm{Pois}(\lambda_j^{(k)})$.

    We can do this in the following way. 
    Let $s = \Exp[u_1(A_1)] = \cdots = \Exp[u_n(A_n)]$.
    We will make each bundle $A_i$ to have expected utility $s_k = \frac{\lfloor k s \rfloor}{k}$.
    Choose one item $g_i \in A_i$ from each bundle. Then for each $j \in M \setminus \{g_1, \ldots, g_n\}$, let $\lambda_j^{(k)} = \frac{\lfloor k \lambda_j \rfloor}{k} \in \mathbb{Q}$. For $g_i$, let $\lambda_{g_i}^{(k)} =  s_k - \sum_{j \in A_i \setminus \{g_i\}} \frac{\lfloor k \lambda_j \rfloor}{k}$.
    We restrict to sufficiently large $k$, so that all approximating parameters are strictly positive.

    Denote $P_k(B)$ and $P(B)$ the ex-post EF probability of an allocation $B$ on stage-$k$ approximation instance and the original instance, resp. 
    We claim that $P_k(B) \to P(B)$ as $k\to \infty$.

    Since each $\lambda_j^{(k)} \to \lambda_j$, for any utility profile $x=(x_{ij})_{i,j}$, the probability of utilities realizing to $x$ in the stage-$k$ instance converges to that of the original instance:
    $$p_k(x)=\prod_{i,j} e^{-\lambda_j^{(k)}}\frac{(\lambda_j^{(k)})^{x_{ij}}}{x_{ij}!} \to \prod_{i,j} e^{-\lambda_j}\frac{\lambda_j^{x_{ij}}}{x_{ij}!} = p(x).$$

    Let $E_B$ be the set of utility profiles under which $B$ is EF. 
    Then we have $P_k(B) = \sum_{x \in E_B}p_k(x)$ and  $P(B) = \sum_{x \in E_B}p(x)$. To change the order of the limit and the (infinite) sum, we use the dominated convergence theorem.

    Since $\lambda_j^{(k)} \to \lambda_j$, there exists some constant $C_j$ such that $\lambda_j^{(k)} \le C_j$. By taking $C = \max_j C_j$, we have $\lambda_j^{(k)} \le C$ for all $j,k$.
    Then, $|p_k(x)| \le \prod_{i,j} \frac{C^{x_{ij}}}{x_{ij}!}$ holds and we have
    $$\sum_{x \in E_B} \prod_{i,j}\frac{C^{x_{ij}}}{x_{ij}!} \le \sum_{x \in \mathbb{Z}_{\ge 0}^{mn}}\prod_{i,j}\frac{C^{x_{ij}}}{x_{ij}!} = e^{Cmn}< \infty.$$
    Hence, we can change the order and 
    $$P_k(B) = \sum_{x \in E_B}p_k(x) \to \sum_{x \in E_B}p(x) = P(B).$$

    Finally, since $P_k(A) \ge P_k(B)$ by the rational case, by taking $k \to \infty$, we have $P(A) \ge P(B)$.
\end{proof}

\subsection{Proof of \Cref{thm:0.930approx}}\label{sec:0.930}
\approxlogconti*
\begin{proof}
    It remains to prove tightness of the approximation ratio for $n=2$.
	Let $c_2 = \frac{1- \log 2}{2}$ and $c_3=c_4 = \frac{1+\log 2}{4}$ so that $c_2, c_3, c_4 \in (0,1)$ and $c_2 + c_3 + c_4 = 1$.
	Fix $0<\epsilon<c_2$.
	Consider the case of 4 items, with the common distributions being
	\begin{align*}
		u_1(1),u_2(1) &\sim \mathrm{Exp}(1), \\ 
		u_1(2),u_2(2) &\sim \mathrm{Unif}[c_2-\epsilon, c_2 + \epsilon], \\ 
		u_1(3),u_2(3) &\sim \mathrm{Unif}[c_3-\epsilon, c_3 + \epsilon], \\ 
		u_1(4),u_2(4) &\sim \mathrm{Unif}[c_4-\epsilon, c_4 + \epsilon].
	\end{align*}
	
	Then all distributions are log-concave. Also note that the allocation $A'= (\{1\}, \{2, 3, 4\})$ is ex-ante EF:
	$$\Exp[u_1(A_1')] = 1 = c_2 + c_3 + c_4 = \Exp[u_1(A_2')] .$$
	Note that $\Delta_{A'}=u_1(A_1')-u_1(A_2') \xrightarrow{d} u_1(1) - 1$ as $\epsilon \to 0$. That is,
	\begin{align*}
		\Pr(u_1(A_1') \ge u_1(A_2')) = \Pr(\Delta_{A'}\ge0)&\to \Pr(u_1(1)\ge1)=e^{-1}, \\
		\Pr(u_2(A_2') \ge u_2(A_1')) =\Pr(\Delta_{A'}\le0)&\to \Pr(u_1(1)\le1)=1-e^{-1},
	\end{align*}
	as $\epsilon \to 0$.
	Hence, the ex-post EF probability of $A'$ converges to $\frac{1}{e} \left( 1 - \frac{1}{e} \right) = \frac{e-1}{e^2}$.
	
	Now consider the allocation $A^*=(\{1,2\},\{3,4\})$. Then, 
	$\Delta_{A^*}=u_1(A_1^*)-u_1(A_2^*) \xrightarrow{d} u_1(1) - \log 2$ as $\epsilon \to 0$, and hence 
	\begin{align*}
		\Pr(u_1(A_1^*) \ge u_1(A_2^*)) = \Pr(\Delta_{A^*}\ge0)&\to \Pr(u_1(1)\ge \log 2)=\frac{1}{2}, \\
		\Pr(u_2(A_2^*) \ge u_2(A_1^*))=\Pr(\Delta_{A^*}\le0)&\to \Pr(u_1(1)\le \log 2)=\frac{1}{2},
	\end{align*}
	as $\epsilon \to 0$. Hence, the ex-post EF probability of $A^*$ converges to $\frac{1}{4}$.
	Thus, the approximation ratio is tight: $\frac{(e-1)/e^2}{1/4}=\frac{4(e-1)}{e^2}$.
\end{proof}

\subsection{Proof of \Cref{thm:strongsharpP}}\label{sec:strongsharpP}
Our main model assumes that utilities $u_i(j)$ are mutually independent. However, there are cases where correlation between these random variables is natural. One extreme would be the case of identical valuations, discussed in \Cref{sec:identical}.

We show that for \emph{general} valuations, where utilities $\{u_i(j)\}_{i,j}$ can be dependent in an arbitrary sense, computing the ex-post EF probability of a given allocation is not expected to be done even in pseudo-polynomial time.

In fact, this holds even when possible dependencies only occur on the same item; $u_1(j), \ldots, u_n(j)$ regarding the same item $j$ can be dependent, but utility vectors $\mathbf{u}(j)= (u_1(j), \ldots, u_n(j))$ are independent across $j$.

Note that the analogous problem size for this case would be defined by specifying the probability distribution of the random vector $\mathbf{u}(j)= (u_1(j), \ldots, u_n(j))$ for each $j$. That is, 
the pmf of the vector is represented explicitly as a tuple $((x_{1,\ell},\ldots,x_{n,\ell}), p_{\ell})_{\ell=1,2,\ldots,s}$ of integer-vectors in the support and its corresponding probability where
\[
\Pr(u_1(j)=x_{1,\ell}, \ldots,u_n(j)=x_{n,\ell})=p_{\ell}.
\]
The input size is the total bit length of all support values and rational probabilities. 

Now we state and prove the result.

\strongsharpP*
\begin{proof}
	We perform reduction from \textsc{$\#$3SAT}. Consider a 3CNF formula $\varphi$ with Boolean variables $y_1, \ldots, y_r$ and clauses $C_1, \ldots, C_s$, in which every clause contains exactly three literals on three distinct variables. We want to compute the number $N(\varphi)$ of assignments that satisfy $\varphi$.
	
	We will construct the fair division of lotteries instance with $s+1$ agents $N=\{0,1, \ldots, s\}$ and $m=r+2s$ items.
	For each clause $C_i$, we create a clause agent $i \in [s]$, and we also create a dummy agent $0$.
	For each Boolean variable $y_j$, consider a uniform random variable $Y_j \sim \mathrm{Unif}\{0,1\}$. Let $Y_1, \ldots, Y_r$ be independent. We will first make $r$ variable items $1, \ldots, r$.
	For an agent $i \in [s]$, define $u_i(j)$ as
	\begin{equation*}
		u_i(j) = 
		\begin{cases*}
			1-Y_j & \text{ if $C_i$ contains the literal $y_j$,}\\
			Y_j & \text{ if $C_i$ contains the literal $\neg y_j$,}\\
			0 & \text{ otherwise.}\\
		\end{cases*}
	\end{equation*}
	For the dummy agent, let $u_0(j) = 0$.
	Observe that for a realization $(Y_1, \ldots, Y_r)=(y_1, \ldots, y_r) \in \{0,1\}^r$, $u_i(\{1, \ldots, r\})= \sum_{j=1}^r u_i(j)$ becomes the number of literals in $C_i$ that evaluate to false.
	
	Next we construct $2s$ clause-guard items. For each clause $C_i$, create two items $g_i', g_i''$ and define $u_i(g_i')=u_i(g_i'')=1$ and $u_k(g_i')=u_k(g_i'')=0$ for $k \ne i$. Note that $u_i(\{g_i', g_i''\}) = 2$ for $i \in [s]$ and $u_i(\{g_k', g_k''\}) = 0$ for $0\ne k \ne i$.

	Now consider the allocation $A$ with $A_0 = \{1, \ldots, r\}$ and $A_i = \{g_i', g_i''\}$ for $i \in [s]$.
	By previous considerations, we have (with probability 1) $u_i(A_i) \ge u_i(A_j)$  for $i,j \in [s]$, and $u_0 \equiv 0$. This makes the probability of $A$ being ex-post EF be
		\[
		\Pr(u_i(A_i) \ge u_i(A_0), \forall i \in [s])  = \Pr(2 \ge u_i(A_0), \forall i \in [s]).
		\]
	
	Since $u_i(A_0)$ is the number of literals in $C_i$ that is evaluated to false by $Y_1, \ldots, Y_r$ as stated above, $2 \ge u_i(A_0)$ is equivalent to $C_i$ being evaluated to true. Hence, 
		\begin{align*}
		\Pr(2 \ge u_i(A_0), \forall i \in [s]) = \Pr(\varphi(Y_1, \ldots, Y_r) = \mathrm{True})= \frac{N(\varphi)}{2^r} ,
		\end{align*}
	and we can recover $N(\varphi)$.
\end{proof}

\end{document}